\documentclass[12pt]{article}
\usepackage{graphicx}
\usepackage{epstopdf}
\usepackage{placeins}
\usepackage[caption=false,font=footnotesize,labelfont=sf,textfont=sf]{subfig}
\usepackage{algpseudocode}
\usepackage{algorithmicx}
\usepackage{algorithm}
\usepackage{url}
\usepackage{amsmath}
\usepackage{amsfonts}
\usepackage{amsthm}
\usepackage{cite}
\usepackage{geometry}

\newtheorem{theorem}{Theorem}
\newtheorem{definition}{Definition}

\usepackage[affil-it]{authblk} 
\usepackage{etoolbox}
\usepackage{lmodern}

\makeatletter
\patchcmd{\@maketitle}{\LARGE \@title}{\fontsize{16}{19.2}\selectfont\@title}{}{}
\makeatother

\title{Spectral clustering algorithms for the detection of clusters in block-cyclic and block-acyclic graphs}
\author{Hadrien Van Lierde and Tommy W. S. Chow\\
	Department of Electronic Engineering, City University of Hong Kong\\
	83 Tat Chee Av., Kowloon Tong, Hong Kong, China\\
	hvanlierd2-c@my.cityu.edu.hk, eetchow@cityu.edu.hk\\
	\and
	Jean-Charles Delvenne\\
	Department of Mathematical Engineering, Universite Catholique de Louvain,\\
	Avenue Georges Lemaitre 4, B-1348 Louvain-la-Neuve, Belgium\\
	jean-charles.delvenne@uclouvain.be\\
	}
\date{}

\begin{document}
\footnotesize\noindent\textit{This is the unrefereed Author's Original Version of the article. A peer-reviewed version has been accepted for publication in the Journal of Complex Networks published by Oxford University Press. The present version is not the Accepted Manuscript.}
{\let\newpage\relax\maketitle}
\maketitle

\begin{abstract}
We propose two spectral algorithms for partitioning nodes in directed graphs respectively with a cyclic and an acyclic pattern of connection between groups of nodes. Our methods are based on the computation of extremal eigenvalues of the transition matrix associated to the directed graph. The two algorithms outperform state-of-the-art methods for directed graph clustering on synthetic datasets, including methods based on blockmodels, bibliometric symmetrization and random walks. Our algorithms have the same space complexity as classical spectral clustering algorithms for undirected graphs and their time complexity is also linear in the number of edges in the graph. One of our methods is applied to a trophic network based on predator-prey relationships. It successfully extracts common categories of preys and predators encountered in food chains. The same method is also applied to highlight the hierarchical structure of a worldwide network  of Autonomous Systems depicting business agreements between Internet Service Providers.\\

\noindent\textbf{keywords:} Complex networks, spectral clustering, cyclic graph, acyclic graph, stochastic blockmodel, directed graph.
%%%% If classification number provided then
\end{abstract}

\section{Introduction}
\label{sec:introduction}
The past years have witnessed the emergence of large networks in various disciplines including social science, biology, and neuroscience. These networks model pairwise relationships between entities such as predator-prey relationships in trophic networks, friendship in social networks, etc. These structures are usually represented as graphs where pairwise relationships are encoded as edges connecting vertices in the graph. When the relationships between entities are not bidirectional, the resulting graph is directed. Some directed networks in real-world applications have a block-acyclic structure: nodes can be partitioned into groups of nodes such that the connections between groups form an acyclic pattern as depicted in figure \ref{figtest1}. Such patterns are encountered in networks that tend to have a hierarchical structure such as trophic networks modelling predator-prey relationships \cite{Elton1927} or networks of Autonomous Systems where edges denote money transfers between Internet Service Providers \cite{Center2013}. On the other hand, one may encounter directed graphs with a block-cyclic structure (figure \ref{figtest4}) when the network models a cyclic phenomenon such as the carbon cycle \cite{Post1990}. These two patterns are intimately related as the removal of a few edges from a block-cyclic graph makes it block-acyclic. This relationship is also observed in real-world networks: a graph of predator-prey interactions can be viewed as an acyclic version of the carbon cycle. In this paper, we take advantage of this connection between the two types of patterns and formulate two closely related algorithms for the detection of groups of nodes respectively in block-acyclic and block-cyclic graphs in the presence of slight perturbations.

The partitioning of nodes in block-acyclic and block-cyclic networks can be viewed as a clustering problem. In graph mining, clustering refers to the task of grouping nodes that are similar in some sense. The resulting groups are called clusters. In the case of directed graphs, the definition of similarity between two nodes may take the directionality of edges incident to these nodes into account. Clustering algorithms  taking the directionality of edges into account may be referred to as \textit{pattern-based clustering} algorithms which extract \textit{pattern-based clusters} \cite{Malliaros2013}: such methods produce a result in which nodes within the same cluster have similar connections with other clusters. Groups of nodes in block-acyclic and block-cyclic graphs are examples of pattern-based clusters.

Several approaches were proposed for the detection of pattern-based clusters in directed graphs \cite{Malliaros2013}. Popular families of methods for the detection of pattern-based clusters are random walk based algorithms, blockmodels and more specifically stochastic blockmodels and bibliometric symmetrization. Random walk based models are usually meant to detect density-based clusters \cite{Chung2005}, however by defining a two step random walk as suggested in \cite{Huang2006} pattern-based clusters such as blocks in block-cyclic graphs can also be detected. But, the success of this method is guaranteed only when the graph is strongly connected and the result is hazardous when the graph is sparse, with a high number of nodes with zero inner or outer degree. Models based on a blockmodelling approach \cite{Reichardt2007} are based on the definition of an image graph representing connections between blocks of nodes in a graph and the block membership is selected so that the corresponding image graph is consistent with the edges of the original graph. However, in existing algorithms the optimization process relies, for instance, on simulated annealing, hence the computational cost is high and there is a risk of falling into a local optimum. Moreover, this method may also fail when the graph is sparse. Clustering algorithms based on stochastic blockmodels detect clusters of nodes that are stochastically equivalent. In particular the method proposed in \cite{Sussman2012} estimates the block membership of nodes by defining a vertex embedding based on the extraction of singular vectors of the adjacency matrix which turns to be efficient compared to the common methods based on expectation maximization. However, the assumption of stochastic equivalence implies that the degrees of nodes within clusters exhibit a certain regularity as shown further. Hence, this approach may yield poor results in detecting clusters in real-world block-cyclic and block-acyclic networks. A related category of method is bibliometric symmetrization which defines a node similarity matrix as a weighted sum between the co-coupling matrix $WW^T$ and the co-citation matrix $W^TW$ \cite{Satuluri2011} where $W$ is the adjacency matrix of the graph. However it may also fail when the degrees of nodes are not sufficiently regular within groups. To relax this assumption, degree corrected versions with variables representing the degrees of nodes were proposed \cite{Karrer2011,Ramasco2008}. But fitting these models relies on costly methods that do not eliminate the risk of falling into a local optimum (simulated annealing, local heuristics, etc.) \cite{Peixoto2014}. Hence methods based on random walks, bibliometric symmetrization, blockmodels with or without degree correction, may yield poor results in the detection of blocks of nodes in block-cyclic and block-acyclic graphs due to assumptions of connectivity or regularity or due to the computational difficulty of solving the associated optimization problems. The methods described in this paper partly alleviate these weaknesses.

In this paper, we present two new clustering algorithms that extract clusters in block-cyclic and block-acyclic graphs. The first algorithm, called \textit{Block-Cyclic Spectral} (BCS) clustering algorithm is designed for the detection of clusters in block-cyclic graphs. The second algorithm, referred to as \textit{Block-Acyclic Spectral} (BAS) clustering algorithm is a slight extension of the first one that is able to detect clusters in block-acyclic graphs. We apply the second algorithm to two real-world datasets: a trophic network in which the traditional classification of agents in an ecosystem is detected, from producers to top-level predators, and a worldwide network of Autonomous Systems depicting money transfers between Internet Service Providers. When tested on synthetic datasets, our algorithms produce smaller clustering errors than other state-of-the-art algorithms. Moreover, our methods only involve standard tools of linear algebra which makes them efficient in terms of time and space complexity.

Hence the approach we follow differs from other clustering methods for directed graphs: we restrict ourselves to two patterns of connection (cyclic and acyclic) but we make no assumption of regularity (for instance on the degrees of nodes). Our proposed algorithms are based on the computation of complex eigenvalues and eigenvectors of a non-symmetric graph-related matrix, commonly called the \textit{transition matrix}. The process is similar to the well-known \textit{Spectral Clustering} algorithm for the detection of clusters in undirected graphs which is also based on the computation of eigenvalues and eigenvectors of a graph-related matrix \cite{Von2007}. However, spectral clustering and extensions of spectral clustering to directed graphs are essentially based on the real spectrum of symmetric matrices associated to the graph \cite{Sussman2012,Gleich2006,Pentney2005}. In contrast, our method is based on the complex spectrum of a non-symmetric matrix. Hence it keeps the intrinsically asymmetric information contained in directed graphs while having approximately the same time and space complexity as other spectral clustering algorithms. A paper recently appeared \cite{Klymko2016} that exploits spectral information (in a different way than in the present paper) for solving a related problem, the detection of block-cyclic components in the communities of a graph, with a special focus on the three-block case. In contrast, we focus on networks with a global block-cyclic structure and extend our method for the detection of acyclic structures, which we deem even more relevant than block-cyclicity in practical situations. 
Part of the results presented here were derived in an unpublished technical report \cite{Van2015}. The present paper offers more empirical validation and comparison with state-of-the-art competing techniques.

\begin{figure*}[!t]
\centering
\subfloat[Block-acyclic graph]{\includegraphics[width=2.5in]{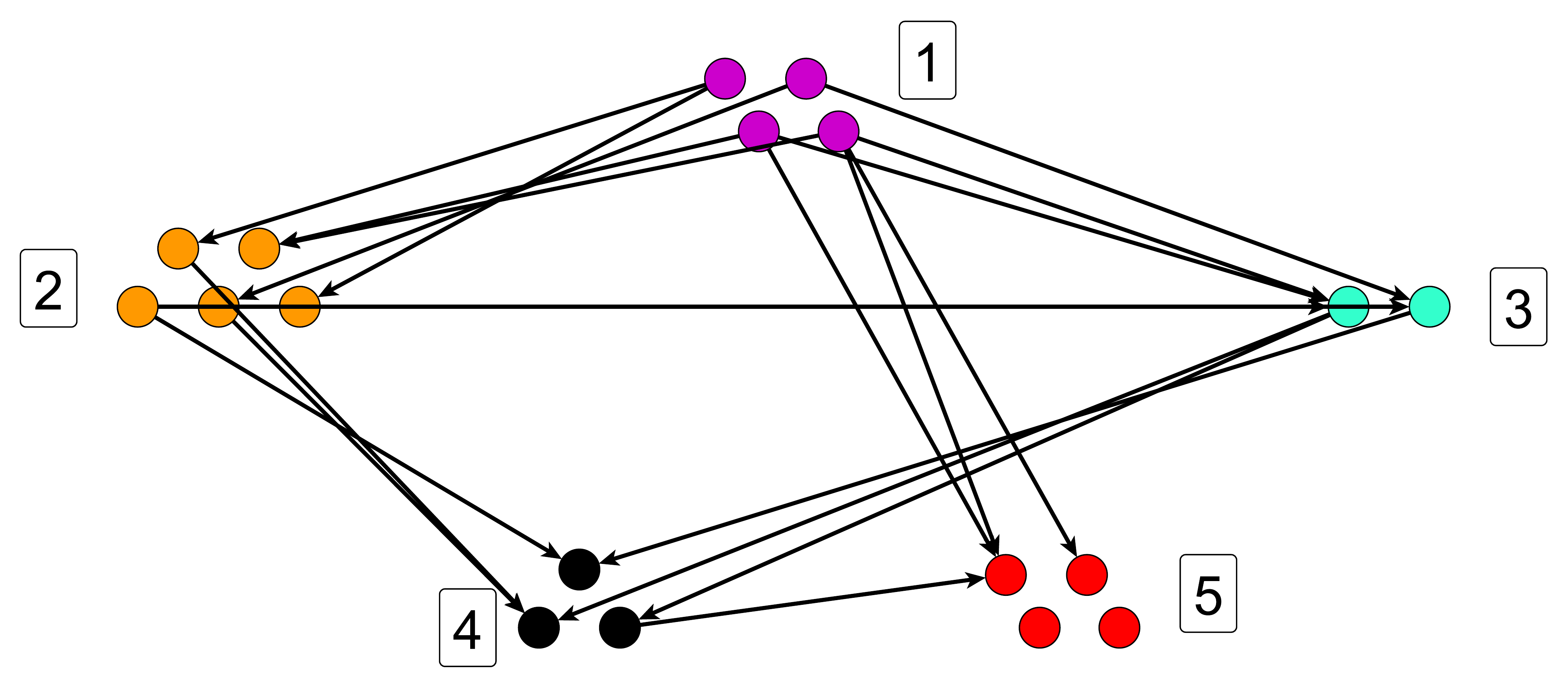}\label{figtest1}}
\hfil
\subfloat[Block-cyclic graph]{\includegraphics[width=2.5in]{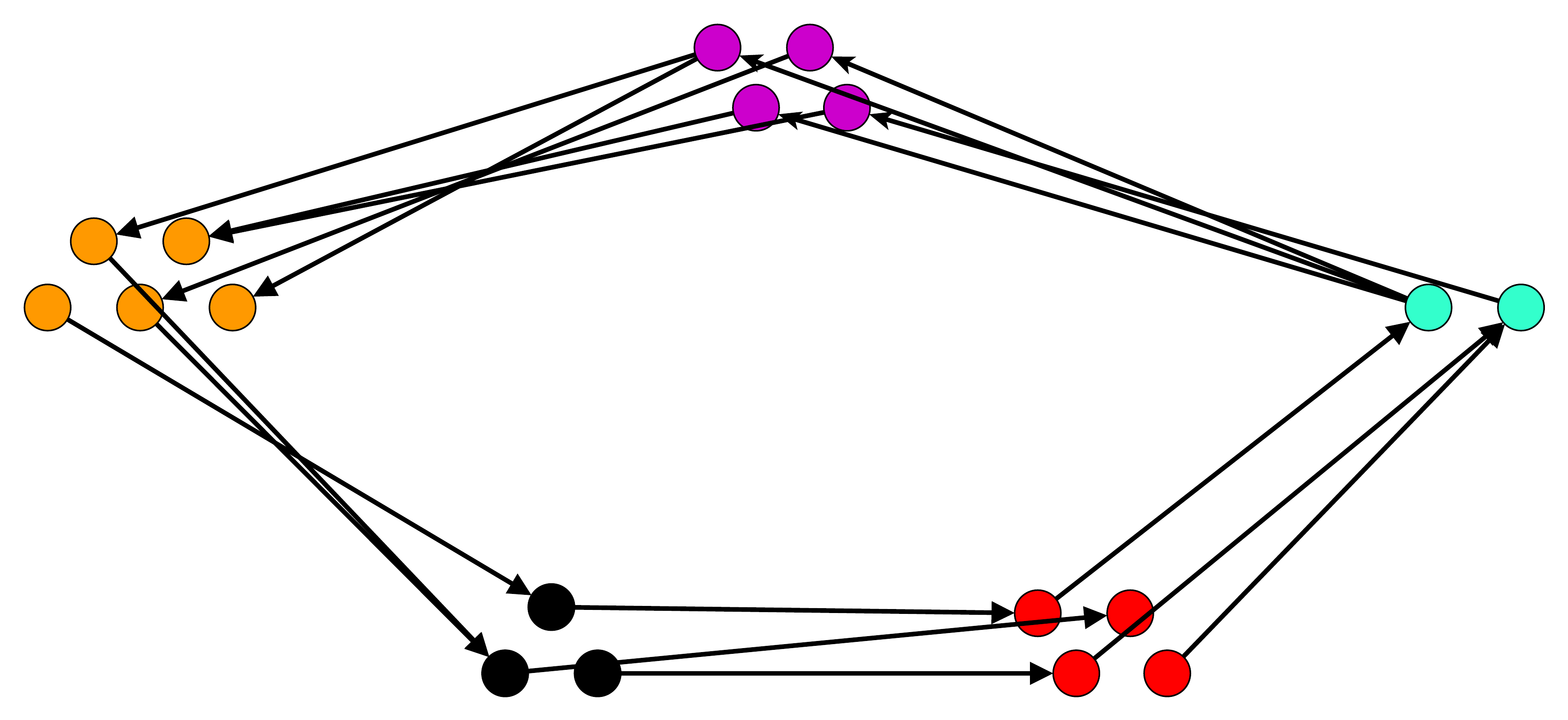}\label{figtest4}}
\caption{Block-acyclic and block-cyclic graphs. Labels of blocks in the block-acyclic graph denote the ranking of blocks (topological order of blocks in the block-acyclic graph).}
\label{fig_sim}
\end{figure*}

The structure of this paper is as follows. In section \ref{sectionRelatedWork}, we describe related clustering methods for directed graphs. In section \ref{BCSsection}, we present our BCS clustering algorithm for the detection of clusters in block-cyclic graphs. Then we describe the links between block-cyclic and block-acyclic graphs in section \ref{nested}. In section \ref{BASsection}, BAS clustering algorithm is introduced for the detection of clusters in block-acyclic graphs. In section \ref{ExperimentSection}, we analyse the performances of BCS and BAS clustering algorithms on synthetic data. Finally, in section \ref{ApplicationSection}, we apply BAS clustering algorithm to a real-world trophic network and a network of Autonomous Systems.

\section{Related work}\label{sectionRelatedWork}
In this section, we present existing algorithms related to our work, including the classical spectral clustering algorithm and some existing algorithms for clustering directed graphs.

\subsection{Spectral clustering of undirected graphs}
Spectral clustering uses eigenvalues and eigenvectors of a graph-related matrix (the \textit{Laplacian}) to detect \textit{density-based clusters} of nodes in an undirected graph, namely clusters with a high number of intra-cluster edges and a low number of inter-cluster edges \cite{Von2007}. The method can be decomposed into two steps. First, the nodes of the graph are mapped to points in a Euclidean space such that nodes that are likely to lie in the same cluster are mapped to points that are close to each other in this projection space. The second step of the method involves clustering the $n$ points in $\mathbb{R}^k$ using k-means algorithm. The algorithm is based on spectral properties of the graph Laplacian $L\in\mathbb{R}^{n\times n}$ defined by
\begin{equation}
L_{uv}=\left\lbrace\begin{array}{ll}
1-\frac{W_{uu}}{d_u}&\text{ if }u=v\text{ and }d_u\neq 0\\
-\frac{W_{uv}}{\sqrt{d_ud_v}}&\text{ if }u\text{ and }v\text{ are adjacent}\\
0&\text{ otherwise}
\end{array}\right.
\end{equation}
where $W$ is the adjacency matrix of the graph and $d_u$ is the degree of node $u$. If the target number of clusters is $k$, extracting the eigenvectors associated to the $k$ smallest eigenvalues of the Laplacian and storing them as the columns of a matrix $U\in\mathbb{R}^{n\times k}$, the embeddings of nodes are given by the $n$ rows of $U$. One can justify this method in the following way. When clusters are disconnected, namely when the graph contains $k$ connected components, the rows of $U$ associated to nodes belonging to the same component are identical \cite{Von2007}. Hence, when perturbing this disconnected graph, namely in the presence of clusters with a low number of inter-cluster edges, nodes within the same clusters are mapped to points that tend to form clusters in $\mathbb{R}^k$. This is explained by the semi-simplicity of eigenvalue $0$ of the graph Laplacian which implies the continuous dependence of associated eigenvectors on the weights of edges in the graph \cite{Von2007}. In sections \ref{BCSsection} and \ref{BASsection}, we show that a similar approach can be used to extract clusters in directed graphs with a cyclic or an acyclic pattern of connection between clusters: we also use the spectrum of a graph-related matrix to map nodes to points in a Euclidean space and cluster these points with k-means algorithm.

\subsection{Clustering algorithms for directed graphs}
In this section, we describe existing clustering algorithms for the detection of pattern-based clusters in directed graphs, namely groups of nodes with similar connections to other groups in some sense. We focus on methods that are theoretically able to extract blocks from block-cyclic and block-acyclic graphs.

Bibliometric symmetrization refers to a symmetrization of the adjacency matrix $W$ of $G$. The symmetrized matrix $(1-\alpha)W^TW+\alpha WW^T$ is defined as the adjacency matrix of an undirected graph $G_u$ for a certain choice of weighing parameter $\alpha$. This symmetric adjacency matrix is a linear combination of the co-coupling matrix $WW^T$ and the co-citation matrix $W^TW$. Then clustering methods for undirected graphs are applied to $G_u$ such as a spectral clustering algorithm. This method is efficient to detect co-citation networks \cite{Satuluri2011}.

The primary focus of random walk based clustering algorithms is the detection of density-based clusters \cite{Chung2005}, namely with a high number of intra-cluster edges and a low-number of inter-cluster edges. A symmetric Laplacian matrix for directed graphs based on the stationary probability distribution of a random walk is defined and applying classical spectral clustering algorithm to this Laplacian matrix leads to the extraction of clusters in which a random walker is likely to be trapped. To detect pattern-based clusters, an extension of this method was proposed in which a random walker alternatively moves forward following the directionality of edges, and backwards, in the opposite direction \cite{Huang2006}. This method successfully extracts clusters in citation-based networks. Similarly, another random walk-based approach extends the definition of directed modularity to extract clusters of densely connected nodes with a cyclic pattern of connections between clusters \cite{Conrad2015,Conrad2016}.

The blockmodelling approach is based on the extraction of functional classes from networks \cite{Reichardt2007}. Each class corresponds to a node in an image graph which describes the functional roles of classes of nodes and the overall pattern of connections between classes of nodes. A measure of how well a given directed graph fits to an image graph is proposed. The optimal partitioning of nodes and image graph are obtained by maximizing this quality measure using alternating optimization combined with simulated annealing.

Methods based on stochastic blockmodels were first defined for undirected networks and then exten-ded to directed graphs in \cite{Wang1987}. A stochastic blockmodel is a model of random graph. For a number $k$ of blocks the parameters of a stochastic blockmodel are a vector of probabilities $\rho\in\{0,1\}^k$ and a matrix $P\in \{0,1\}^{k\times k}$. Each node is randomly assigned to a block with probabilities specified by $\rho$ and the probability of having an edge $(i,j)$ for $i$ in block $s$ and $j$ in block $t$ is $P_{st}$. For this reason, nodes within a block are said to be stochastically equivalent. One of the various methods for the detection of blocks in graphs generated by a stochastic blockmodel is based on the extraction of singular vectors of the adjacency matrix \cite{Sussman2012}, which is similar to the bibliometric symmetrization combined with the classical spectral clustering algorithm. The common definition of the stochastic blockmodel implies that in- and out-degrees of nodes within blocks follow a Poisson binomial distribution \cite{Karrer2011,Le1960} and have thus the same expected value. As this assumption is not verified in most real-world directed networks, \cite{Peixoto2014} proposed a degree-corrected stochastic blockmodel for directed graphs where additional variables are introduced allowing more flexibility in the distribution of degrees of nodes within blocks. The partitioning of nodes is obtained by an expectation maximization process.

Other statistical methods exist among which the so-called clustering algorithm for content-based networks \cite{Ramasco2008}. This method is similar to stochastic blockmodelling but instead of block-to-block probabilities of connection, it is based on node-to-block and block-to-node probabilities. The model parameters are adjusted through an expectation maximization algorithm. This approach can be viewed as another degree-corrected stochastic blockmodel and hence it is robust to high variations in the degrees of nodes but it also involves a more complex optimization approach due to the higher number of parameters.

Finally, some methods are based on the detection of roles in directed networks such as in \cite{Beguerisse2013} which defines the number of paths of given lengths starting or ending in a node as its features from which node similarities are extracted. As we will see, our definition of block-cyclic and block-acyclic graph does not include any constraint on the regularity of node features such as the number of incoming or outgoing paths.

We are interested in the detection of clusters in block-cyclic and block-acyclic graphs. Apart from the role model, the methods described above are all theoretically able to extract such clusters. Methods based on bibliometric symmetrization and stochastic blockmodels are able to detect such structures whenever the assumption of stochastic equivalence between nodes within blocks is verified. Provided that the graph is strongly connected, the method based on two-step random walk can also be used. If degrees of nodes are large enough, the blockmodelling approach is also successful. However, the benchmark tests presented in section \ref{ExperimentSection} show that our algorithms outperform all these methods in the presence of high perturbations or when these assumptions are not fulfilled.

\section{Spectral clustering algorithm for block-cycles}\label{BCSsection}
In this section, we describe a method for the extraction of blocks of nodes in block-cyclic graphs (or block-cycles). We recall that a block-cycle is a directed graph where nodes can be partitioned into non-empty blocks with a cyclic pattern of connections between blocks. We provide a formal definition of block-cycle below. In subsequent sections, we refer to weighted directed graphs as triplet $G=(V,E,W)$ where $V$ is a set of nodes, $E\subseteq V\times V$ is a set of directed edges and $W\in\mathbb{R}^{n\times n}_+$ is a matrix of positive edges weights. When the graph is unweighted, we refer to it as a pair $G=(V,E)$.

\begin{definition}[Block-cycle]
A directed graph $G=(V,E,W)$ is a block-cycle of $k$ blocks if it contains at least one directed cycle of length $k$ and if there exists a function $\tau:V\rightarrow\{1,...,k\}$ partitioning the nodes of $V$ into $k$ non-empty subsets, such that
\begin{equation}
E\subseteq\{(u,v)\text{ : }(\tau(u),\tau(v))\in \mathcal{C}\}
\end{equation}
where $\mathcal{C}=\{(1,2),(2,3),...,(k-1,k),(k,1)\}$.
\end{definition}
Due to the equivalence between the existence of \textit{clusters} in a graph and the \textit{block} structure of the adjacency matrix, we use the terms "cluster" and "block" interchangeably. We also use the terms "block-cycle" and "block-cyclic graph" interchangeably. Figure \ref{figtest4} displays an example of block-cycle. Blocks may contain any number of nodes (other than zero) and there can be any number of edges connecting a pair of consecutive blocks in the cycle. It is worth mentioning that, in the general case, a given block-cycle is unlikely to derive from a stochastic blockmodel; in which nodes within a block are stochastically equivalent \cite{Sussman2012}. Indeed, as mentioned before, stochastic equivalence implies that degrees of nodes within the same block are identically distributed. Our definition does not include such regularity assumption.

The definition implies that any block-cycle is k-partite \cite{Balakrishnan2012}. However, the converse is not true as the definition of block-cycle includes an additional constraint on the directionality of edges. Similarly, one can view a block-cycle as a generalization of bipartite graph.

Up to a permutation of blocks, the adjacency matrix of a block-cycle is a block circulant matrix with nonzero blocks in the upper diagonal and in the bottom-left corner as depicted in figure \ref{figtest2}. Given a perturbed block-cycle, our goal is to recover the partitioning of nodes into blocks, namely to provide an estimation $\tilde{\tau}$ of $\tau$. To detect blocks in a block-cycle, we use the spectrum of a graph-related matrix, the transition matrix $P\in\mathbb{R}^{n\times n}$ associated to the Markov chain based on the graph.
\begin{equation}
P_{ij}=\left\lbrace\begin{array}{ll}
\frac{W_{ij}}{d_i^{out}}&\text{ if }d_i^{out}\neq 0\\
0&\text{ otherwise}
\end{array}\right.
\end{equation}
where $W$ is the weight matrix and $d_i^{out}=\sum_jW_{ij}$ is the out-degree of node $i$. The basic spectral property of the transition matrix is that all its complex eigenvalues lie in the ball $\{x\in\mathbb{C}\text{ : }\Vert x\Vert_2\leq 1\}$ regardless of the associated graph \cite{Leskovec2014}. This property combined with the fact that the transition matrix of a block-cycle is block circulant \cite{Mazancourt1983} makes it possible to prove the following theorem\footnote{The proof of the theorem is provided in appendix.}. We make the assumption that $d_i^{out}>0$ for any node $i$.

\begin{figure*}[!t]
\centering
\subfloat[Adjacency matrix]{\includegraphics[width=2in]{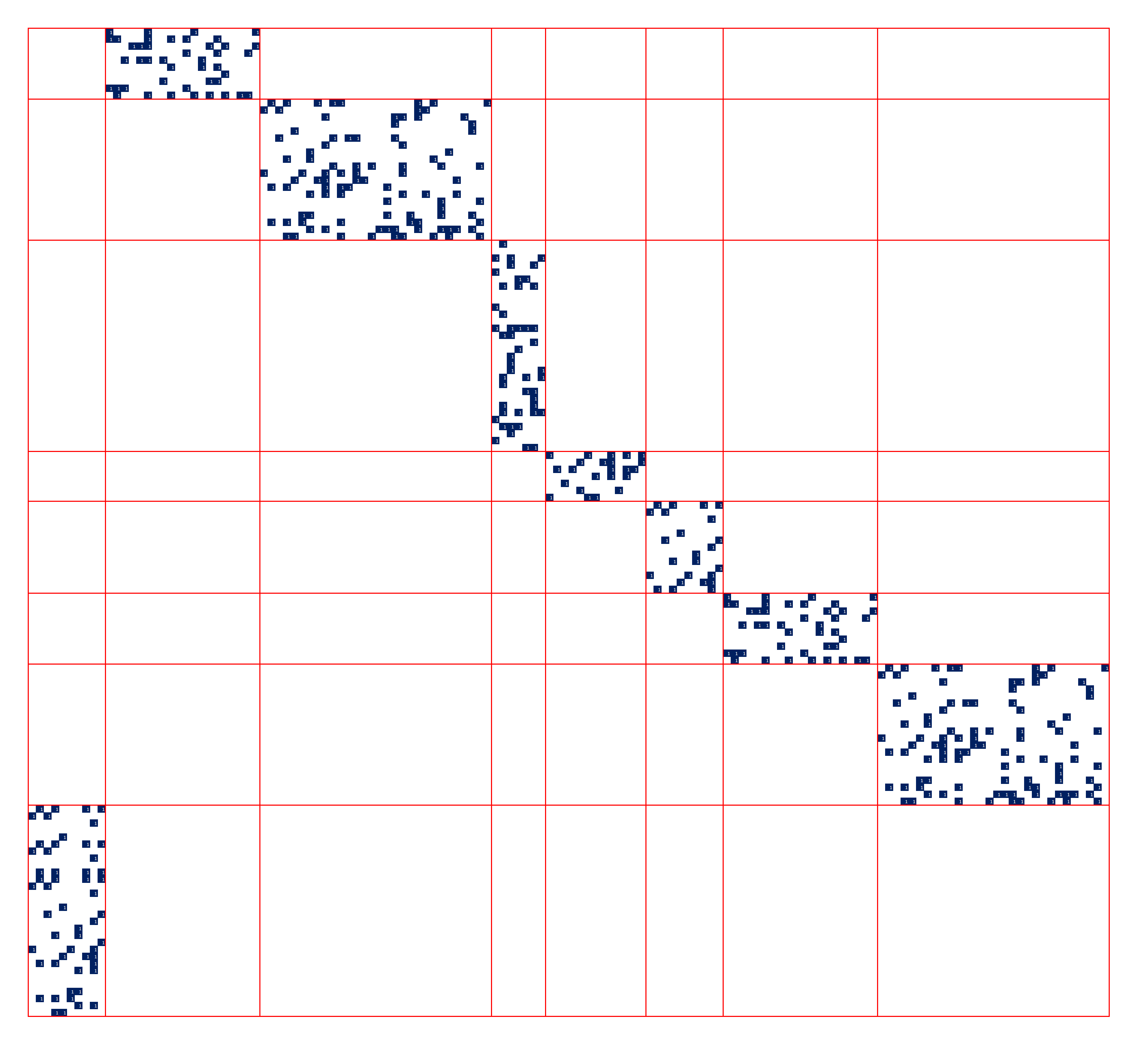}\label{figtest2}}
\hfil
\subfloat[Spectrum of transition matrix]{\includegraphics[width=2.8in]{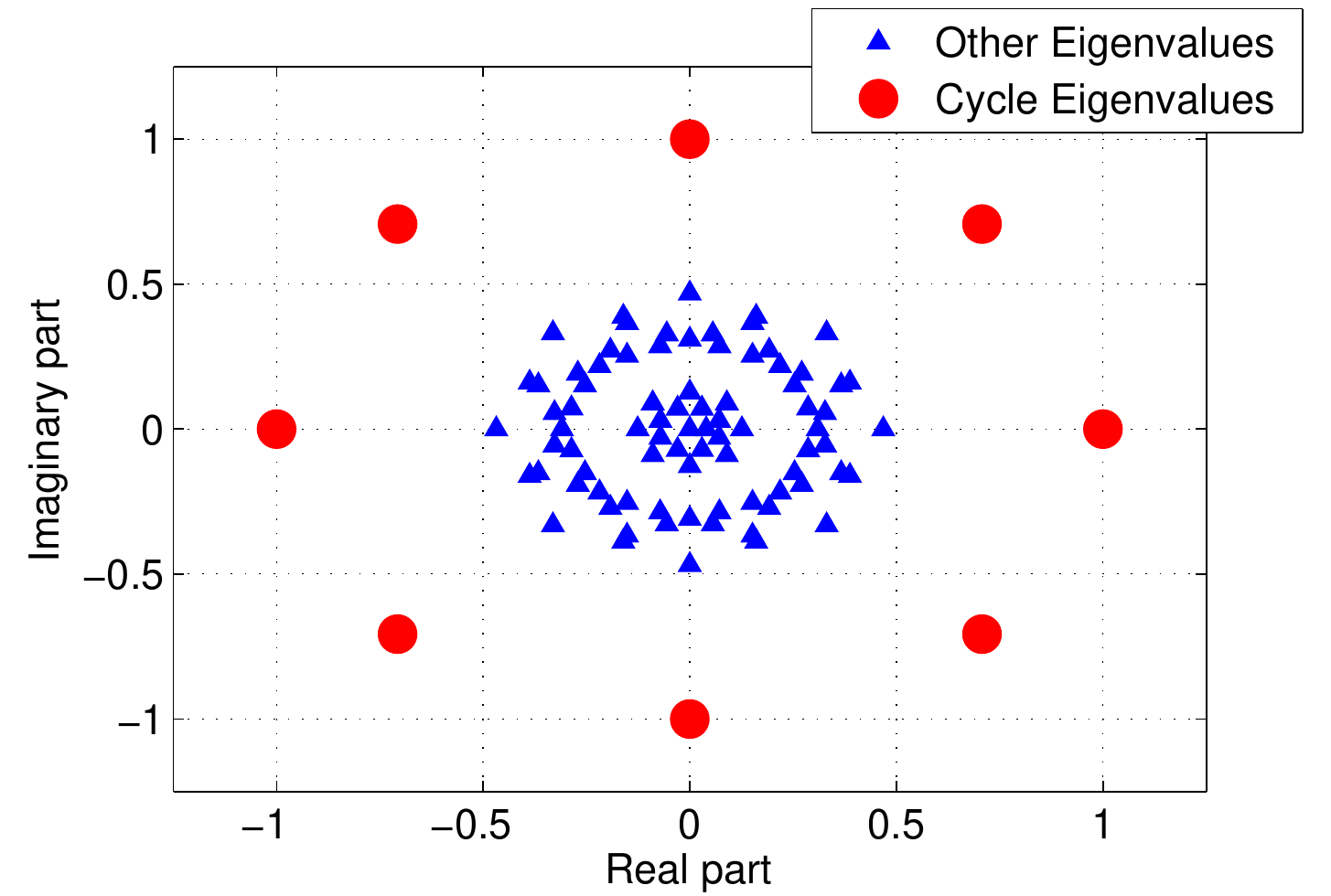}\label{figtest3}}
\caption{Adjacency matrix (left) and complex spectrum of the transition matrix (right) of a block-cycle of $8$ blocks.}
\label{fig_block_cyclic}
\end{figure*}

\begin{theorem}\label{theoremCycle}
Let $G=(V,E,W)$ be a block-cycle with $k$ blocks $V_1,...,V_k$ such that $d_i^{out}>0$ for all $i\in V$. Then $\lambda_l=e^{-2\pi i\frac{l}{k}}\in spec(P)$ for all $0\leq l\leq k-1$, namely there are $k$ eigenvalues located on a circle centered at the origin and with radius $1$ in the complex plane. The eigenvector associated to the eigenvalue $e^{-2\pi i\frac{l}{k}}$ is
\begin{equation}
u^l_j=\left\{
\begin{array}{ll}
e^{2\pi i\frac{lk}{k}} & j\in V_1\\
e^{2\pi i\frac{l(k-1)}{k}} & j\in V_2\\
\vdots & \\
e^{2\pi i\frac{l}{k}} & j\in V_k
\end{array}.\right.
\end{equation} 
Moreover, if $G$ is strongly connected, then the eigenvalues $\lambda_0,...,\lambda_{k-1}$ have multiplicity $1$ and all other eigenvalues of $P$ have a modulus strictly lower than $1$.
\end{theorem}

We refer to these eigenvalues and eigenvectors as the \textit{cycle eigenvalues} and \textit{eigenvectors}. This spectral property is illustrated in figure \ref{figtest3} displaying the eigenvalues of the transition matrix of a block-cycle of eight blocks. Hence, in order to recover the $k$ blocks of a block-cycle, one may compute the cycle eigenvalues and store the corresponding eigenvectors as the columns of a matrix $U\in\mathbb{C}^{n\times k}$. Rows of $U$ corresponding to nodes within the same block are identical. In the case of an unperturbed block-cycle, Depth First Search is clearly faster at recovering blocks but it might fail in the presence of perturbation. Hence, following an approach similar to classical spectral clustering algorithm, we propose a clustering method based on cycle eigenvalues and eigenvectors that effectively recovers blocks of nodes in a block-cycle in the presence of slight perturbations. In the presence of perturbation, this method extracts eigenvalues that are close to the unperturbed cycle eigenvalues, stores the corresponding eigenvectors as the columns of a matrix $U\in\mathbb{C}^{n\times k}$ and clusters rows of this matrix to recover blocks. Theorem \ref{theoremPerturbation} justifies this approach by quantifying the effect of additive perturbations on cycle eigenvalues and cycle eigenvectors: starting with an unperturbed block-cycle $G$, we consider a graph $\hat{G}$ obtained by appending perturbing edges to $G$ and we analyse its spectral properties\footnote{The proof of the theorem is provided in appendix.}.

\begin{theorem}\label{theoremPerturbation}
Let $G=(V,E,W)$ be a strongly connected block-cycle with $k$ blocks $V_1,...,V_k$ such that $d_i^{out}>0$ for all $i\in V$, let $\lambda_0,...,\lambda_{k-1}$ be the $k$ cycle eigenvalues and $u^0,...,u^{k-1}$ be the corresponding cycle eigenvectors. Let the $\hat{G}=(V,\hat{E},\hat{W})$ be a perturbed version of $G$ formed by appending positively weighted edges to $G$ except self-loops. Let $P$ and $\hat{P}$ denote the transition matrices of $G$ and $\hat{G}$ respectively. We define the quantities
\begin{equation}\label{deltaKetc}
\begin{array}{rcl}
\sigma &=&\underset{(i,j)\in\hat{E}}{\max}\text{ }\frac{\hat{d}_j^{in}}{\hat{d}_i^{out}}\\
\rho &=&\underset{i}{\max}\text{ }\frac{\hat{d}_i^{out}-d_i^{out}}{d_i^{out}}
\end{array}
\end{equation}

where $d^{in}_i$, $d^{out}_i$, $\hat{d}^{in}_i$ and $\hat{d}^{out}_i$ represent the in-degree and out-degree of $i$-th node in $G$ and $\hat{G}$ respectively. Then,
\begin{enumerate}
\item for any cycle eigenvalue $\lambda_l\in spec(P)$, there exists an eigenvalue $\hat{\lambda}_l\in spec(\hat{P})$ so that
\begin{equation}
\left\vert\hat{\lambda}_l-\lambda_l\right\vert \leq \sqrt{2n}\Vert f\Vert_2\sigma^{\frac{1}{2}}\rho^{\frac{1}{2}}+\mathcal{O}\left(\sigma\rho\right)
\end{equation}
where $f$ is the Perron eigenvector of $P$, namely the left eigenvector of the transition matrix of $G$ associated to eigenvalue $1$ with positive entries and $\Vert f\Vert_1=1$,

\item there exists an eigenvector $\hat{u}^l$ of $\hat{P}$ associated to eigenvalue $\hat{\lambda}^l$ verifying
\begin{equation}\label{ineq2}
\Vert\hat{u}^l-u^l\Vert_2\leq \sqrt{2}\Vert (\lambda^lI-P)^{\#}\Vert_2 \sigma^{\frac{1}{2}}\rho^{\frac{1}{2}}+\mathcal{O}\left(\sigma\rho\right)
\end{equation} 
where $u^l$ is the eigenvector of $P$ associated to eigenvalue $\lambda^l$ and $(\lambda^lI-P)^{\#}$ denotes the Drazin generalized inverse of $(\lambda^lI-P)$.
\end{enumerate}
\end{theorem}

We make a few comments about the perturbation bounds above. $f$ is the Perron eigenvector \cite{Seneta2006} with unit $1$-norm. Thus $\frac{1}{\sqrt{n}}\leq \Vert f\Vert_2\leq 1$ and $\Vert f\Vert_2=\frac{1}{\sqrt{n}}$ when it is constant namely when the stationary probability of the random walk associated to $P$ is uniform over vertices. $\rho$ measures the relative discrepancy between out-degrees in the presence and in the absence of perturbation while $\sigma$ measures the discrepancy between in-degree of destination and out-degree of origin for all edges in the perturbed graph. In particular, the perturbation bound is small if the block-cycle has homogeneous degrees and if the perturbation is uniform. However, experiments described in section \ref{ExperimentSection} show that our method is robust to variations in node degrees and non-uniform perturbations. Regarding the perturbation of eigenvectors, inequality \ref{ineq2} follows from the fact that cycle eigenvalues are simple. Providing bounds on the norm of the Drazin inverse of a non-symmetric matrix is tedious since it depends on the condition number of each eigenvalue of the matrix (see for instance \cite{Kirkland2012} for bounds on the norm of $(I-P)^{\#}$ for a stochastic matrix $P$). However, an important factor impacting the norm $\Vert (\lambda_lI-P)^{\#}\Vert_2$ of the Drazin inverse of $(\lambda_lI-P)$ is the inverse of its smallest nonzero eigenvalue\cite{Stewart1990}
\begin{equation}
\left(\underset{\lambda\in \text{spec}(P)\setminus\{\lambda_l\}}{\min}|\lambda-\lambda_l|\right)^{-1}
\end{equation}
namely the inverse of the minimum distance between cycle eigenvalues and other eigenvalues of the transition matrix which is typically greater than $0.1$ in real-world applications presented in section \ref{ApplicationSection}. In general, equation \ref{ineq2} implies that the cycle eigenvectors of a block-cycle vary continuously as functions of the entries of the transition matrix and hence of the edges' weights. Although the bounds provided by theorem \ref{theoremPerturbation} can be quite loose in practice, the continuity of cycle eigenvalues and eigenvectors is verified for any strongly connected block-cycle. This continuity property provides the theoretical foundation of our spectral clustering algorithm. It is worth mentioning that the eigenvectors of interest in classical spectral clustering also vary continuously as functions of entries of the Laplacian matrix of an undirected graph which also provides a theoretical justification for classical spectral clustering \cite{Von2007}.

Given the format of cycle eigenvectors (theorem \ref{theoremCycle}), the bound on the perturbation of cycle eigenvalues and the continuity property of the entries of cycle eigenvectors (theorem \ref{theoremPerturbation}), when extracting eigenvectors associated to the $k$ eigenvalues with largest modulus in a slightly perturbed block-cycle of $k$ blocks, entries of these eigenvectors tend to form clusters corresponding to blocks of nodes. This property was verified experimentally including on block-cycles with heterogeneous degrees (see section \ref{ExperimentSection}). Hence, to recover the partitioning of nodes of a perturbed block-cycle, we may compute the $k$ eigenvalues of the transition matrix with largest modulus, store the corresponding eigenvectors as the columns of a matrix $U\in\mathbb{C}^{n\times k}$ and cluster its rows using k-means algorithm. We define algorithm \ref{algorithmBCS} as the \textit{Block-Cyclic Spectral} (BCS) clustering algorithm.

\begin{algorithm}
\caption{Block-Cyclic Spectral (BCS) clustering algorithm}
~\\
\textbf{Input: } Adjacency matrix $W\in\{0,1\}^{n\times n}$ in which all nodes have nonzero out-degree;\\
\textbf{Parameters: }$k\in\{2,3,...,n\}$;\\
\textit{Step 1:} Compute the transition matrix $P$;\\
\textit{Step 2:} Find the $k$ cycle eigenvalues (the $k$ eigenvalues with largest modulus) and store the associated cycle eigenvectors as the columns of a matrix $\Gamma\in \mathbb{C}^{n\times k}$.;\\
\textit{Step 3:} Consider each row of $\Gamma$ as a point in $\mathbb{C}^k$ and cluster these points using a k-means algorithm. Let $\phi:\{1,...,n\}\rightarrow\{1,...,k\}$ be the function assigning each row of $\Gamma$ to a cluster;\\
\textit{Step 4:} Compute the estimation of block membership function $\tilde{\tau}$: $\tilde{\tau}(u)=\phi(u)$ for all $u\in\{1,...,n\}$;\\
\textbf{Output: }estimation of block membership $\tilde{\tau}$;\\
\label{algorithmBCS}
\end{algorithm}

We now make some observations about BCS clustering algorithm. Step 2 of the algorithm involves the computation of eigenvalues of a non-symmetric matrix. As we seek extremal eigenvalues, this can be done with Arnoldi algorithm \cite{Saad1992}. If $k$ is sufficiently small so that we can neglect the time of computation of the spectrum of a $k\times k$ matrix, the time complexity of Arnoldi algorithm is $\mathcal{O}(k^2|E|)$ where $|E|$ is the number of nonzero elements in matrix $W$ \cite{Lee2009}. The k-means method at step 3 can be Lloyd's algorithm \cite{Lloyd1982} the time complexity of which is $\mathcal{O}(k^2nI)$ where $I$ is the maximum allowed number of iterations. In practice, the convergence of Lloyd's algorithm is fast in comparison with that of Arnoldi algorithm. Hence, the overall complexity of BCS clustering algorithm is $\mathcal{O}(k^2|E|)$ in practice, which is linear in the number of edges in the graph.

Eigenvalue $1$ is among the eigenvalues of largest modulus extracted at step 2 of the algorithm. From theorem \ref{theoremCycle}, the associated eigenvector is constant and does not provide any information about the blocks. Moreover other cycle eigenvalues and associated eigenvectors form complex conjugate pairs hence providing redundant information. In practice, after applying Arnoldi algorithm we may keep eigenvalues with positive imaginary part only and exclude eigenvalue $1$ in order to speed up k-means algorithm in step 3.

From theorem \ref{theoremCycle}, $e^{\frac{2\pi i}{k}}$ is an eigenvalue of the transition matrix of a block-cycle and the components of the associated eigenvector are $u_j=e^{2\pi i\frac{l-1}{k}}\text{, if node }j\text{ is in block }l$. Hence, an alternative to algorithm \ref{algorithmBCS} is to extract this eigenvector only and cluster its components in $\mathbb{C}$ to recover the block membership of nodes. However, algorithm \ref{algorithmBCS} turns to be more robust to perturbations for a similar time complexity when the number $k$ of blocks is small. Indeed, as cycle eigenvalues are extremal eigenvalues, the cost of Arnoldi algorithm does not differ significantly for the extraction of one or all cycle eigenvalues.

We leave $k$ as a parameter of the algorithm. In the trophic network and the network of Autonomous Systems presented in section \ref{ApplicationSection}, the number of blocks is known (based on general knowledge about the networks). But in some cases, the number of blocks may not be known in advance. In the conclusion of this paper, we discuss future works including the development of an automatic method to find the suitable number of blocks.

\section{Spectral properties of nested block-cycles}\label{nested}
In this section we provide an empirical analysis of an extension of block-cycles. The spectral properties of so-called nested block-cycles provide the basis for a clustering algorithm for block-acyclic graphs presented in section \ref{BASsection}. The formal definition of a nested block-cycle is given below.
\begin{definition}[Nested block-cycle]
A directed graph $G=(V,E,W)$ is a nested block-cycle of $k$ blocks if there exists a function $\tau:V\rightarrow\{1,...,k\}$ partitioning the nodes of $V$ into $k$ non-empty blocks, such that
\begin{equation}
E\subseteq\{(u,v)\text{ : }\tau(u)<\tau(v)\text{ or }\tau(u) = k\}.
\end{equation}
\end{definition}
An example of nested block-cycle of four blocks is given in figure \ref{nested_example}. In such graph, the $l$-th block may be connected to blocks $l+1,...,k$ for $l<k$ and the $k$-th block may be connected to all other blocks.

We next provide an empirical analysis of the spectrum of the transition matrix of a nested block-cycle and verify to what extent properties of pure block-cycles are preserved in nested block-cycles. As a nested block-cycle may be aperiodic, its spectrum no longer contains $k$ eigenvalues with unit modulus. However, experiments allowed us to make the following observations: when a nested block-cycle of $k$ blocks contains a block-cycle of $k$ blocks as a subgraph that covers all nodes and that verifies all assumptions of theorem \ref{theoremCycle} then
\begin{enumerate}
\item there are $k$ outlying eigenvalues in the spectrum of the transition matrix of $G$, namely $k$ eigenvalues with a significantly larger modulus compared to other eigenvalues,
\item the corresponding eigenvectors verify the same clustering property as cycle eigenvectors of a block-cycle.
\end{enumerate}
To support these observations, we present the following experimental results. Starting with a pure block-cycle in which nodes are randomly assigned to blocks and the probability of existence of edges in the block-cycle is $0.1$, we randomly append edges satisfying the definition of a nested block-cycle which transforms the block-cycle into a nested block-cycle as the one depicted in figure \ref{nested_example}. Figure \ref{eigNested1} displays the evolution of the spectrum of the transition matrix as edges are appended to a block-cycle of four blocks to create a nested block-cycle (we display all the eigenvalues of all transition matrices of the resulting nested block-cycles in the same figure). Although the spectrum is strongly perturbed compared to a pure block-cycle, there are four outlying eigenvalues regardless of the magnitude of the perturbation which confirms the first claim. By analogy with the block-cyclic case, we refer to these eigenvalues as \textit{cycle eigenvalues}. Figure \ref{NBCpert} shows the proportion of misclassified nodes when applying BCS clustering algorithm to the resulting nested block-cycle as a function of the number of appended edges. We also display the proportion of misclassified nodes when the same number of edges are randomly appended. While the error seems to grow linearly for random perturbations, the error stays close to zero for the nested block-cycle no matter the magnitude of the perturbation which supports the second claim above.

\begin{figure}[!t]
\centering
\includegraphics[width=2.5in]{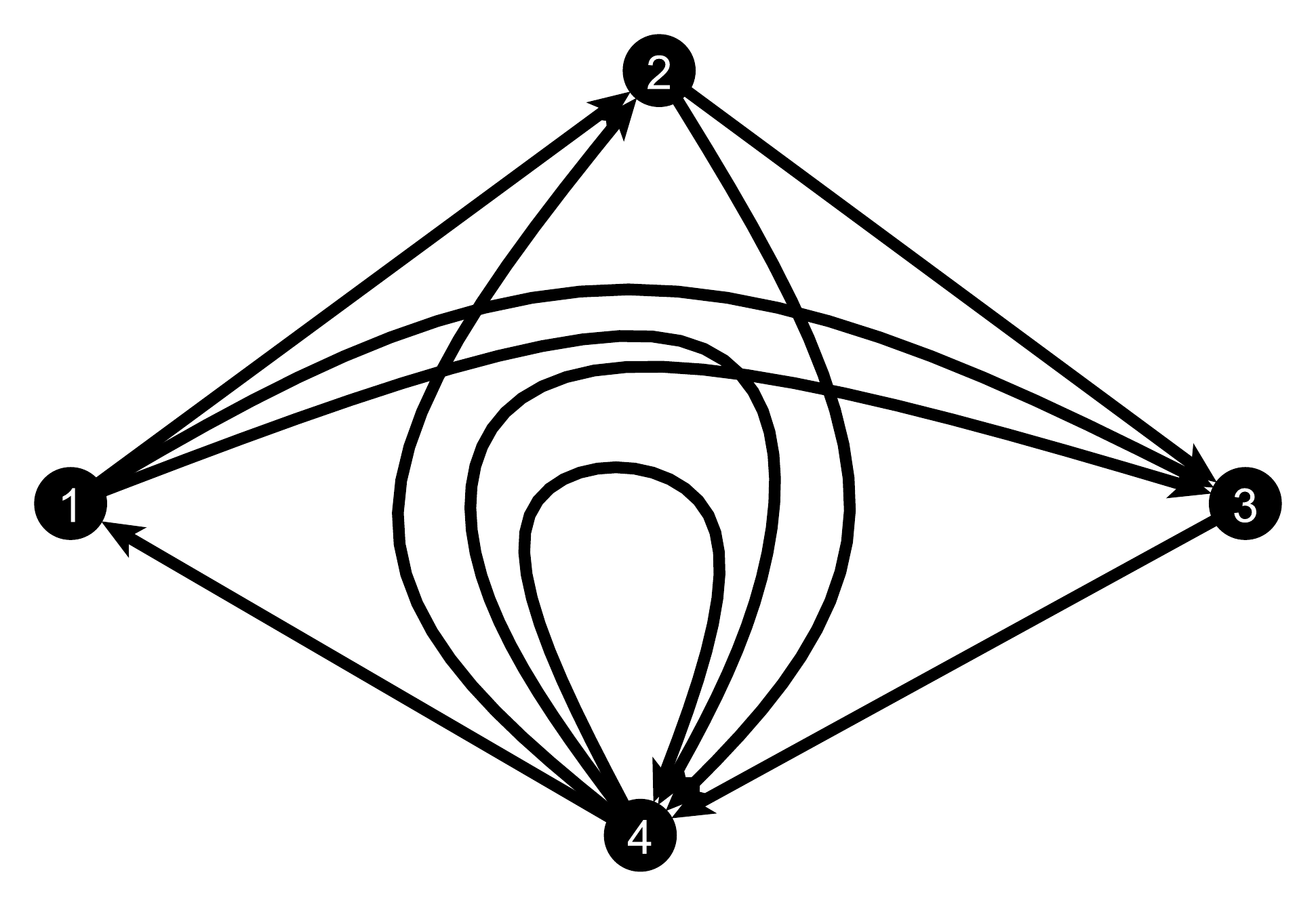}
\caption{Nested block-cycle of four blocks (nodes represent blocks).}
\label{nested_example}
\end{figure}

\begin{figure}[!t]
\centering
\includegraphics[width=2.5in]{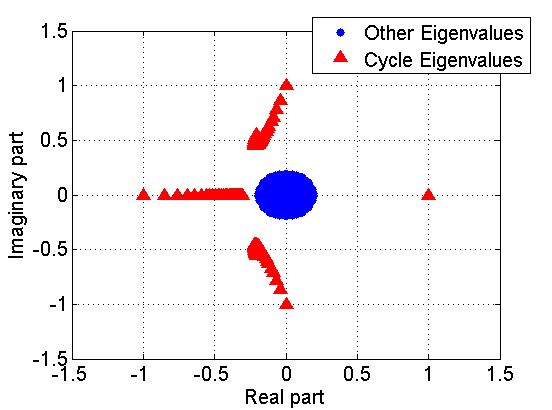}
\caption{Eigenvalues of the transition matrices of nested block-cycles of $4$ blocks obtained by appending edges iteratively to a block-cycle of $4$ blocks. The eigenvalues of all the resulting nested block-cycles are displayed together in the graph above. In each case, there are four eigenvalues with significantly large modulus.}
\label{eigNested1}
\end{figure}

\begin{figure}[!t]
\centering
\includegraphics[width=2.5in]{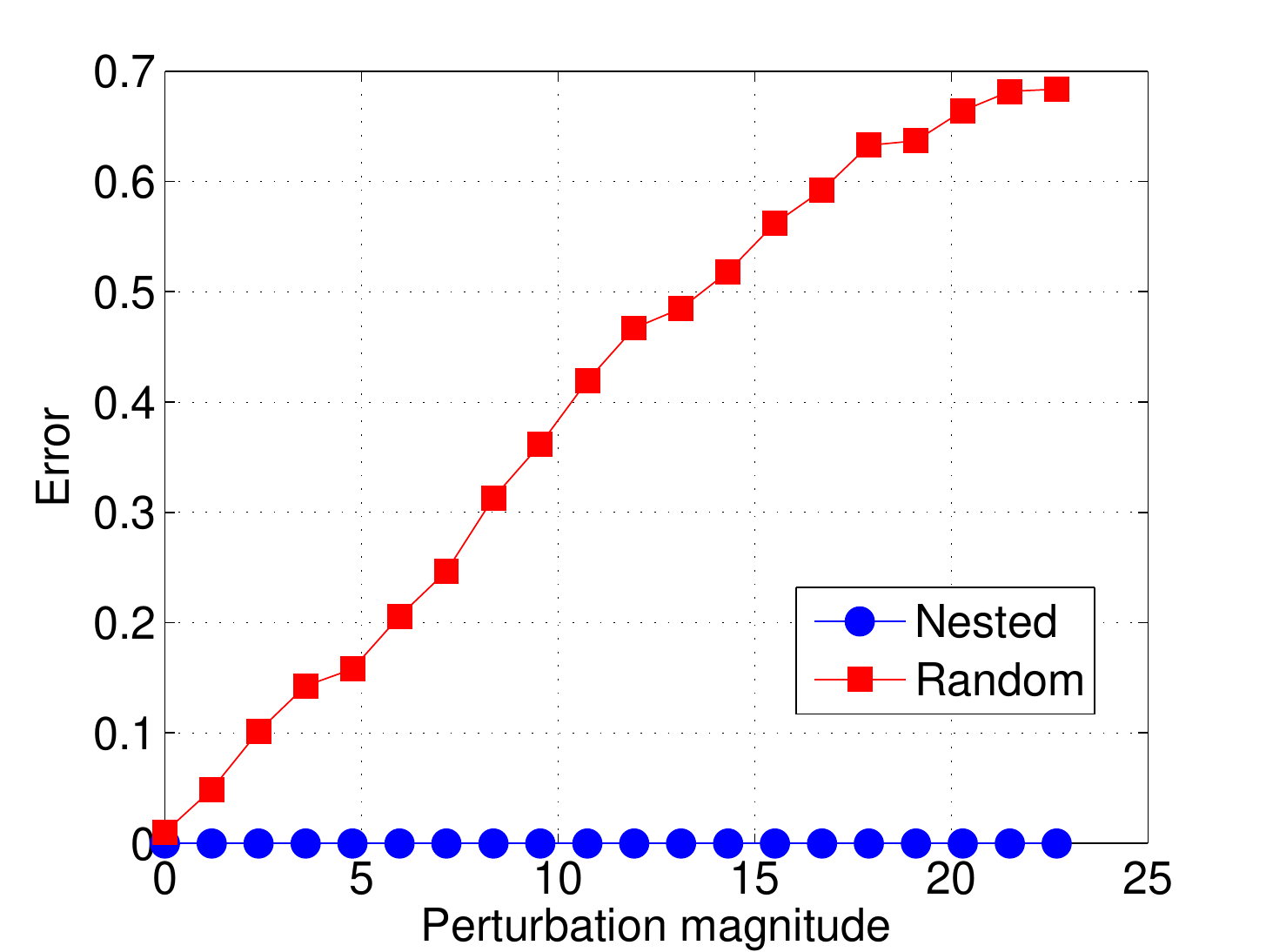}
\caption{Evolution of the proportion of misclassified nodes as a function of the proportion of perturbing edges when applying BCS clustering algorithm respectively to a nested block-cycle and a randomly perturbed block-cycle of $4$ blocks. The proportion of perturbing edges is the ratio between the number of appended edges and the number of edges in the original block-cycle.}
\label{NBCpert}
\end{figure}

\section{Spectral clustering algorithm for block-acyclic graphs}\label{BASsection}

In this section, we describe an algorithm for the extraction of blocks of vertices in block-acyclic graphs. This algorithm referred to as Block-Acyclic Spectral (BAS) clustering algorithm is based on the spectral properties of nested block-cycles described in section \ref{nested} and it is largely similar to BCS clustering algorithm. A block-acyclic graph is a directed graph where vertices can be partitioned into blocks with an acyclic pattern of connections between the blocks. We provide a formal definition of block-acyclic graph below.\\

\begin{definition}[Block-acyclic graph]
A directed graph $G=(V,E,W)$ is a block-acyclic graph of $k$ blocks if there exists a function $\tau:V\rightarrow\{1,...,k\}$ partitioning the nodes of $V$ into $k$ non-empty blocks, such that
\begin{equation}
E\subseteq\{(u,v)\text{ : }\tau(u)<\tau(v)\}.
\end{equation}
\end{definition}

The block membership function $\tau$ can be viewed as a \textit{ranking function} such that any edge of the graph has its origin in a block of strictly lower rank than its destination. Figure \ref{figtest1} displays an example of block-acyclic graph. The adjacency matrix of a block-acyclic graph is strictly block upper triangular as depicted in figure \ref{figtest5}. As in the case of a block-cycle, it is worth mentioning that a block-acyclic graph is unlikely to derive from a stochastic blockmodel as our definition does not include any regularity requirement on the degrees of nodes within blocks.

Our goal is to detect blocks of nodes in block-acyclic graphs. The principle of our method is to append a few edges in order to transform the block-acyclic graph into a nested block-cycle with the same blocks of nodes. We then apply BCS algorithm to this graph to recover the blocks. The following transformation is performed on the block-acyclic graph. If a node is out-isolated (zero out-degree), we artificially add out-edges connecting this node to all other nodes of the graph. In this way, we append edges connecting the block of highest rank back to all other blocks. The resulting graph is a nested block-cycle and, as shown in section \ref{nested}, BCS clustering algorithm is able to recover blocks of nodes in such graph.

Again, we make use of the transition matrix of this modified graph, we denote this matrix by $P_a$.
\begin{equation}
(P_a)_{ij}=\left\{
\begin{array}{ll}
\frac{1}{d_i^{out}}W_{ij} & \text{ if }d_i^{out}>0\\
\frac{1}{n} & \text{ otherwise }
\end{array}.\right.
\end{equation}
We note that this matrix is the transpose of the Google matrix with zero damping factor \cite{Brin2012}. Matrix $P_a$ is the transition matrix of a nested block-cycle. Hence, as mentioned in section \ref{nested}, it has a spectral property similar to the one highlighted in the case of block-cycles: provided that, after appending out-edges to out-isolated nodes, the graph is strongly connected and provided that each node in block $l<k$ is connected to at least one node in block $l+1$, then $k$ eigenvalues of the transition matrix have a modulus significantly larger than the moduli of other complex eigenvalues, as shown in figure \ref{figtest6}. Based on this observation, we formulate BAS clustering algorithm (algorithm \ref{basclusteringalgorithmhe}) similar to BCS clustering algorithm. Both algorithms have approximately the same time and space complexity. As in the case of BCS clustering algorithm we may keep eigenvalues with positive imaginary part only and exclude eigenvalue $1$ at step 2 of BAS clustering algorithm in order to speed up k-means algorithm in step 3.\\

\begin{algorithm}
\caption{Block-Acyclic Spectral (BAS) clustering algorithm}
~\\
\textbf{Input: } Adjacency matrix $W\in\{0,1\}^{n\times n}$;\\
\textbf{Parameters: }$k\in\{2,3,...,n\}$;\\
\textit{Step 1:} Compute transition matrix $P_a$;\\
\textit{Step 2:} Find the $k$ cycle eigenvalues (the $k$ eigenvalues with largest modulus) and store the associated cycle eigenvectors as the columns of a matrix $\Gamma\in \mathbb{C}^{n\times k}$;\\
\textit{Step 3:} Consider each row of $\Gamma$ as a point in $\mathbb{C}^k$ and cluster these points using a k-means algorithm. Let $\phi:\{1,...,n\}\rightarrow\{1,...,k\}$ be the function assigning each row of $\Gamma$ to a cluster;\\
\textit{Step 4:} Compute the estimation of block membership function $\tilde{\tau}$: $\tilde{\tau}(u)=\phi(u)$ for all $u\in\{1,...,n\}$;\\
\textbf{Output: }estimation of block membership $\tilde{\tau}$
\label{basclusteringalgorithmhe}
\end{algorithm}

\begin{figure*}[!t]
\centering
\subfloat[Adjacency matrix]{\includegraphics[width=2in]{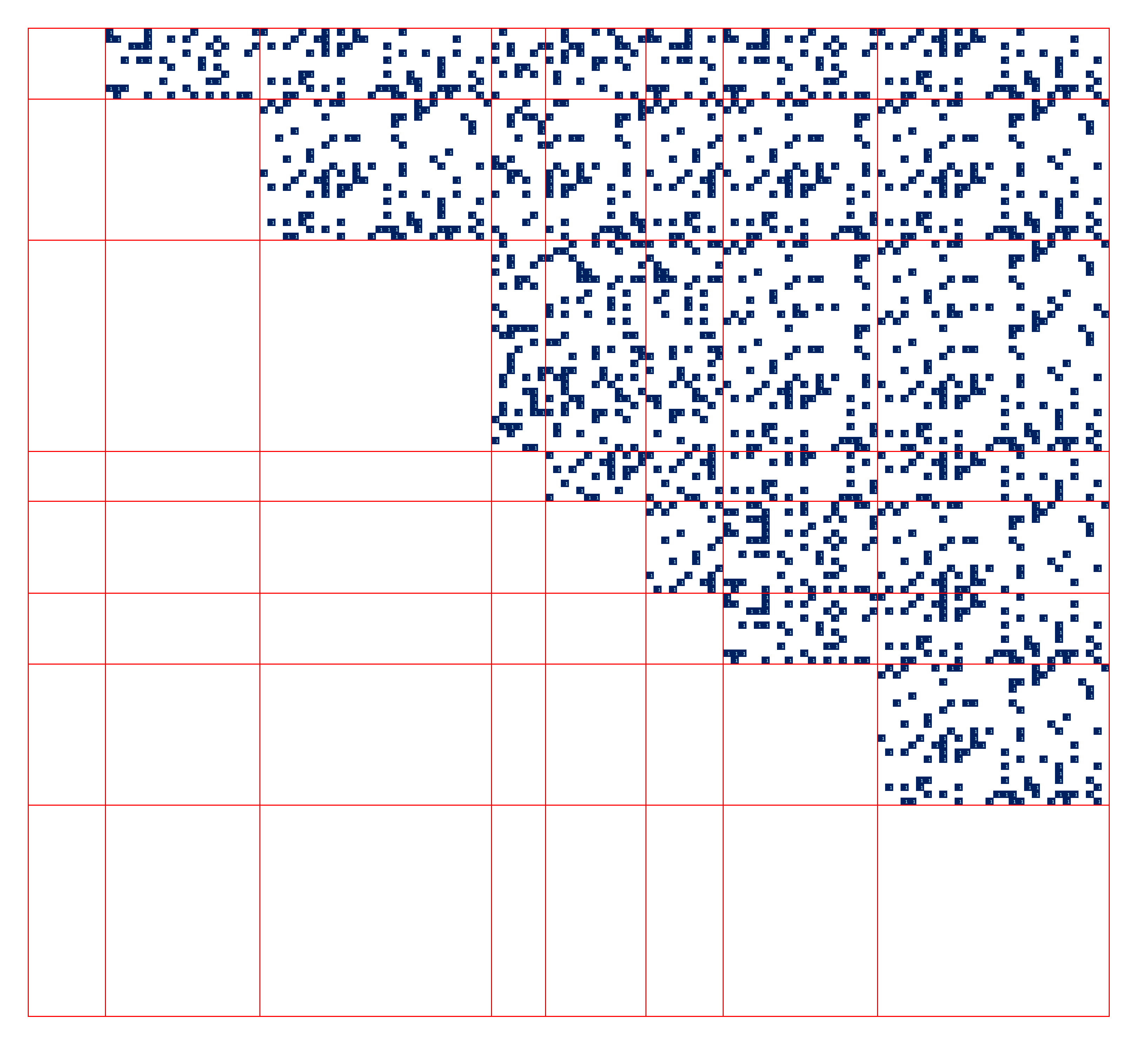}\label{figtest5}}
\hfil
\subfloat[Spectrum of transition matrix]{\includegraphics[width=2.8in]{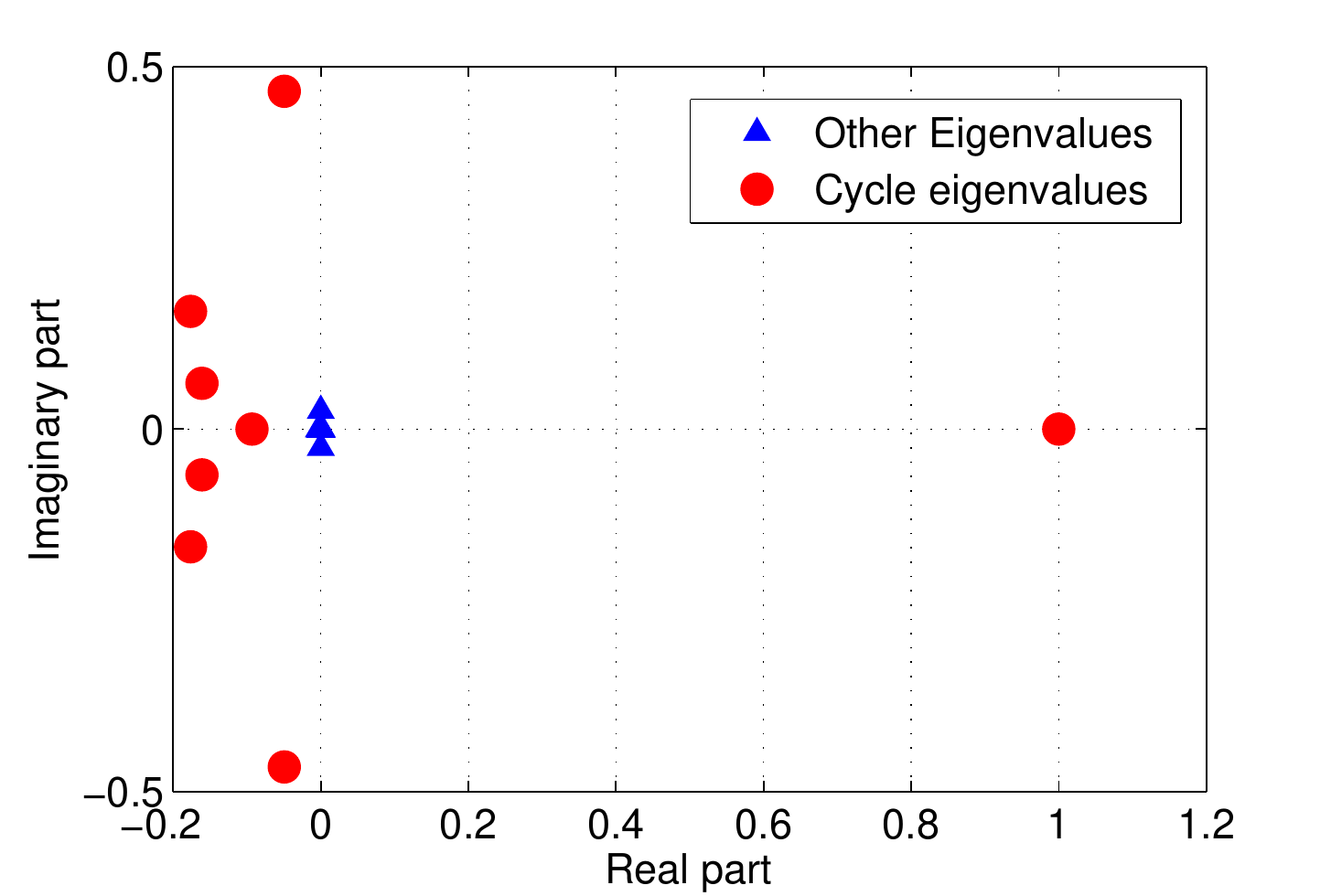}\label{figtest6}}
\caption{Adjacency matrix (left) and complex spectrum of the transition matrix (right) of a block-acyclic graph of $8$ blocks.}
\label{fig_block_acyclic}
\end{figure*}

\section{Experiments}\label{ExperimentSection}

We compare BCS and BAS clustering algorithms to clustering algorithms described previously: bibliometric symmetrization algorithm combined with undirected spectral clustering \cite{Satuluri2011} (BIB), an algorithm based on singular vectors of the adjacency matrix for the estimation of the underlying stochastic blockmodel \cite{Sussman2012} (SVD), an expectation maximization algorithm for degree-corrected stochastic blockmodel \cite{Karrer2011} (DEG), two step random walk algorithm \cite{Huang2006} (RW), Reichardt's method of blockmodelling \cite{Reichardt2007} (REI) and another statistical method called content-based clustering algorithm which is closely related to degree-corrected stochastic blockmodel \cite{Ramasco2008} (CB). Regarding parameter setting, BIB and RW algorithms are both based on slight variations of spectral clustering the main parameter of which is the number $k$ of clusters. For BIB algorithm, we set parameter $\alpha$ to $0.5$ since both co-coupling and co-citation matrices are expected to reflect the block-cyclic (or block-acyclic) shape of the graph. For DEG, REI and CB algorithms, the number of clusters is also provided as an input. Clusters are randomly initialized and an alternating minimization algorithm is applied to optimize an objective function (likelihood for DEG and CB and quality criterion for REI) until no further improvement is achieved. Finally, SVD clustering algorithm relies on the computation of top $d$ right and left singular vectors of the adjacency matrix and the application of k-means algorithm to the components of these vectors. More details about the choice of the target dimension $d$ will be given further. We only consider unweighted graphs but our observations are valid with positive edge weights.

\subsection{Benchmark model}

The graphs used for testing our algorithms are based on stochastic blockmodels. We conduct two experiments both on BCS and BAS clustering algorithms. For the first experiment, two stochastic blockmodels are defined to generate two graphs: a block-cycle and a perturbing graph. Edges from both graphs are combined to form a perturbed block-cycle. As said earlier in this paper, a stochastic blockmodel is a model of random graph in which nodes within blocks are stochastically equivalent \cite{Sussman2012}. We formulate the following mathematical definition of stochastic blockmodel.

\begin{definition}[Stochastic Blockmodel]
Given positive integers $n$ and $k\leq n$ and parameters $\rho\in [0,1]^k$ and $P\in [0,1]^{k\times k}$, an unweighted random graph $G=(V,E)$ of $n$ nodes is generated by the stochastic blockmodel $SBM(k,\rho,P)$ with block membership function $\tau:\{1,...,n\}\rightarrow \{1,...,k\}$ if
\begin{itemize}
\item $P[\tau(u)=i]=\rho_i$ $\forall u\in V$, $i\in\{1,...,k\}$,
\item $P[(u,v)\in E|\tau(u)=s,\tau(v)=t]=P_{st}$ $\forall u,v\in V$.
\end{itemize}
\end{definition}
When considering a graph $G=(V,E)$ generated by the stochastic blockmodel $SBM(k,\rho,P)$ with random block membership function $\tau$, we use the notation $[G,\tau]\sim SBM(k,\rho,P)$. The graphs used in our first experiment are based on the combination of two random graphs: an unperturbed graph $G=(V,E)$ with block membership function $\tau$ such that $[G,\tau]\sim SBM(k,\rho,P)$ where parameter $P$ is chosen so that $G$ is either a block-cycle or a block-acyclic graph and a perturbing graph $\tilde{G}=(V,\tilde{E})$ with block membership function $\tilde{\tau}$ such that $[\tilde{G},\tilde{\tau}]\sim SBM(\tilde{k},\tilde{\rho},\tilde{P})$. We combine $G$ and $\tilde{G}$ to created a perturbed graph $H=(V,E\bigcup \tilde{E})$ and apply BCS or BAS clustering algorithm to provide an estimation $\eta$ of the block membership function $\tau$. Parameter $\tilde{P}$ determines the distribution of perturbing edges in graph $H$. We combine two stochastic blockmodels in such way to take into account the fact that in real-world applications, several complex models and phenomena may influence the shape of networks. This first test applied to both BCS and BAS clustering algorithms is intended to show how our algorithms perform on a standard network model in the presence of perturbation. Our second experiment highlights a case in which our algorithms are successful while other algorithms produce very poor results. We generate two different graphs based on the same stochastic blockmodel $[G_1=(V_1,E_1),\tau_1]\sim SBM(k,\rho,P)$ and $[G_2=(V_2,E_2),\tau_2]\sim SBM(k,\rho,P)$ where $P$ is chosen so that the graphs generated are both block-cyclic or both block-acyclic. Then we combine both graphs by defining $G=(V_1\bigcup V_2,E_1\bigcup E_2\bigcup E_{c})$ where edges in $E_{c}$ are randomly selected according to the model
\begin{equation}\label{modeladv}
P[(u,v)\in E_c]=\left\lbrace
\begin{array}{ll}
\alpha P_{\tau_1(u),\tau_2(v)} & \text{ if }u\in V_1\text{, }v\in V_2\\
\alpha P_{\tau_2(u),\tau_1(v)} & \text{ if }u\in V_2\text{, }v\in V_1\\
0 & \text{ otherwise }
\end{array}
\right.
\end{equation}
where $\alpha\in [0,1]$. If $\alpha$ is equal to $1$, then model (\ref{modeladv}) corresponds to a standard stochastic blockmodel. If $\alpha<1$, the block shape of the adjacency matrix is preserved, but nodes within blocks are not stochastically equivalent. This model illustrates the case where, although the graph might have a block-cyclic or a block-acyclic shape, nodes within a block have a highly different distribution of edges towards  other blocks and hence are not stochastically equivalent. This type of framework is also observed in real-world datasets such as trophic networks modelling predator-prey relationships. For instance, two top-level predators are both at the top of the food chain but they might have highly different diets and have different preys at different levels of the food chain.

Finally, we use the two following error measures to assess the quality of the result \cite{Sussman2012}.
\begin{definition}[Block membership error]
Given a graph $G=(V,E)$ with a certain block membership function $\tau:V\rightarrow\{1,...,k\}$ and a perturbed version $H=(V,E_H)$ of graph $G$, if $\eta:V\rightarrow\{1,...,k\}$ is an estimation of the block membership function $\tau$ based on $H$, then the block membership error is
\begin{equation}\frac{1}{|V|}\underset{\pi\in\Pi(k)}{\min}|\{u\in V\text{ : }\tau(u)\neq \pi(\eta(u))\}|
\end{equation}
where $\Pi(k)$ is the set of permutations on $\{1,...,k\}$.
\end{definition}
In other words, the block membership error computes the minimum number of differences in the entries of $\tau$ and $\eta$ when considering all permutations of block labels $\{1,...,k\}$. Computation of the block membership error can be formulated as a minimum matching problem which can be solved by Hungarian algorithm \cite{Kuhn1955}. We also compute an error based on the \textit{normalized mutual information} \cite{Knops2006} defined below. The error measure associated to the NMI is $1-NMI(\Omega,C)$.
\begin{definition}[Normalized Mutual Information]
Given ground-truth clusters $C=\{c_1,...,c_k\}$ and estimations $\Omega=\{\omega_1,...,\omega_k\}$, the normalized mutual information is defined as

\begin{equation}
NMI(\Omega,C)=\frac{2I(\Omega,C)}{H(\Omega)+H(C)}
\end{equation}
where $I(\Omega,C)$ is the mutual information between $\Omega$ and $C$ and $H(\Omega)$ and $H(C)$ are the entropies respectively of $\Omega$ and $C$.
\end{definition}

\begin{figure*}[!t]
\centering
\subfloat[Block-membership error]{\includegraphics[width=2.5in]{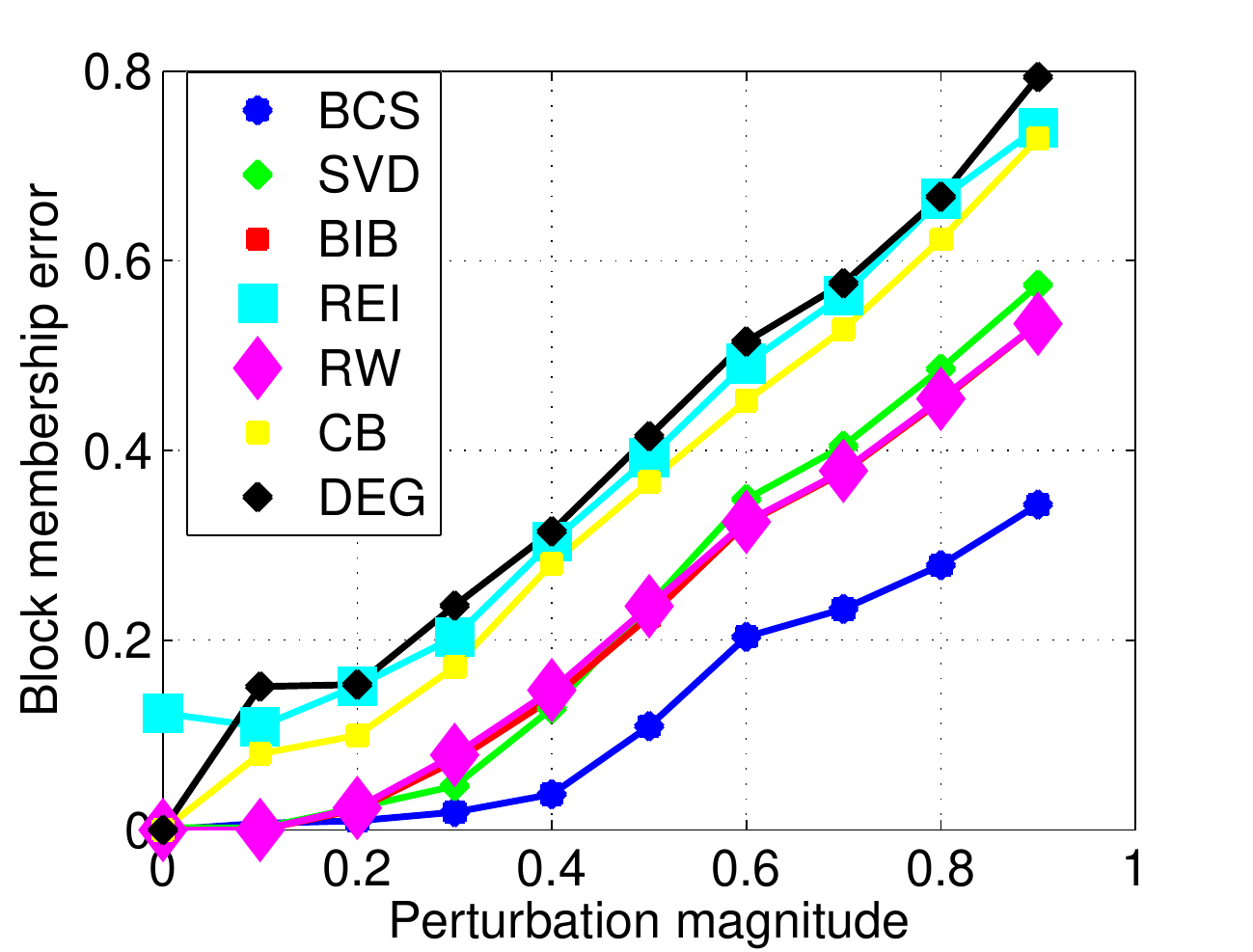}
\label{perturbedBlockCycleBME}}
\hfil
\subfloat[1-NMI]{\includegraphics[width=2.5in]{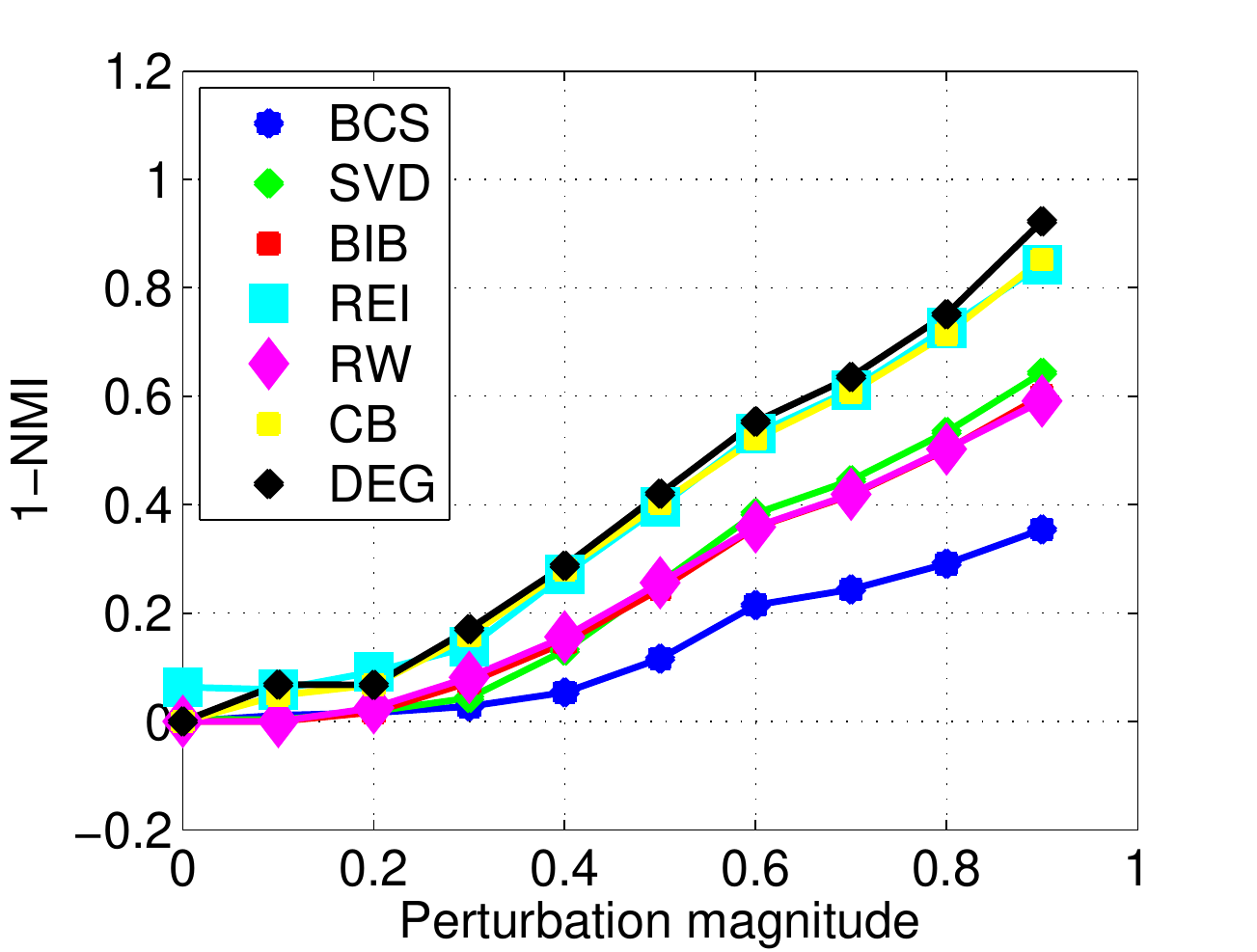}
\label{perturbedBlockCycleNMI}}
\caption{Comparison of the block membership errors (left) and errors based on NMI (right) of our BCS algorithm and SVD, BIB, RW, REI, CB and DEG clustering algorithms as a function of the magnitude of the perturbation (parameter $\epsilon$) on a block-cyclic graph of $n=1000$ nodes.}
\label{perturbedBlockCycle}
\end{figure*}

\subsection{Experiments for BCS clustering algorithm}
For our first experiment, we generate an unperturbed block-cycle $G$ with block membership function $\tau$ such that $[G,\tau]\sim SBM(k,\rho,P)$ with $k=8$ and
\begin{equation}
\begin{array}{rcl}
\rho &=& [0.18,0.2,0.05,0.2,0.14,0.04,0.07,0.13]\\
P_{st}&=&\left\lbrace\begin{array}{ll}
0.7 & \text{ if }(s,t)\in\{(1,2),(2,3),...,(7,8),(8,1)\}\\
0 & \text{ otherwise }
\end{array}\right.
\end{array}.
\end{equation}

where entries of $\rho$ were chosen randomly for obtaining unbalanced blocks of nodes. We generate a perturbing graph $\tilde{G}$ with block membership function $\tilde{\tau}$ such that $[\tilde{G},\tilde{\tau}]\sim SBM(k,\rho,\tilde{P})$ with parameter $\tilde{P}$ of the form $\tilde{P}=\epsilon Q$ where $\epsilon\in [0,1]$ controls the magnitude of the perturbation and $Q\in [0,1]^{k\times k}$ with random entries in $[0,1]$. We compare the partitioning computed by BCS clustering algorithm to the one returned by all other algorithms that were mentioned. Regarding SVD clustering algorithm, the target dimension $d$ is chosen as the rank of matrix $P$ ($k$ in this case) in accordance with the original form of SVD clustering algorithm \cite{Sussman2012}. Figure \ref{perturbedBlockCycle} shows the evolution of the block membership error and the NMI-based error as a function of the perturbation magnitude $\epsilon$ for $n=1000$ nodes. As $\epsilon=1$ would completely hide the block-cyclic structure of the graph, we restrict ourselves to $\epsilon\in [0,0.9]$. Firstly, we observe that BCS clustering algorithm produces an error that is lower than errors obtained with other algorithms, both in terms of block membership error and NMI. Secondly, we see that even in the presence of a $80\%$ perturbation, the block membership error achieved by BCS clustering algorithm is below $30\%$ while all other algorithms produce an error above $40\%$. Our BCS clustering algorithm is thus more robust in the presence of perturbation.

For our second experiment, we generate two graphs following the same stochastic blockmodel described above $SBM(k,\rho,P)$, each containing $500$ nodes. Then we combine them using equation \ref{modeladv} with parameter $\alpha=0.1$. Table \ref{tableExperiment2_1} shows the result obtained. As expected, our method perfectly recovers the block membership of nodes while other methods all produce block membership errors above $30\%$. As in previous experiment, parameter $d$ of SVD clustering algorithm is set to $k$ (rank of matrix $P$).
\begin{table}[!h]
\begin{center}
\begin{tabular}{|| c |c|c||}
   \hline
   Method & Block Membership Error & 1-NMI\\
   \hline
	BCS & $\mathbf{0}$ & $\mathbf{0}$\\
	SVD & $0.5$ & $0.47$\\
	BIB & $0.49$ & $0.47$\\
	REI & $0.34$ & $0.25$\\
	RW & $0.49$ & $0.46$\\
	CB & $0.40$ & $0.49$\\
	DEG & $0.61$ & $0.58$\\   
   \hline
\end{tabular}
\end{center}
\caption{Experiment based on two block-cycles of $500$ nodes generated by stochastic blockmodel $SBM(P,\rho,k)$ and combined to form a block-cycle of $1000$ nodes based on equation \ref{modeladv} with $\alpha=0.1$.}
\label{tableExperiment2_1}
\end{table}

\begin{figure*}[!t]
\centering
\subfloat[Block-membership error]{\includegraphics[width=2.5in]{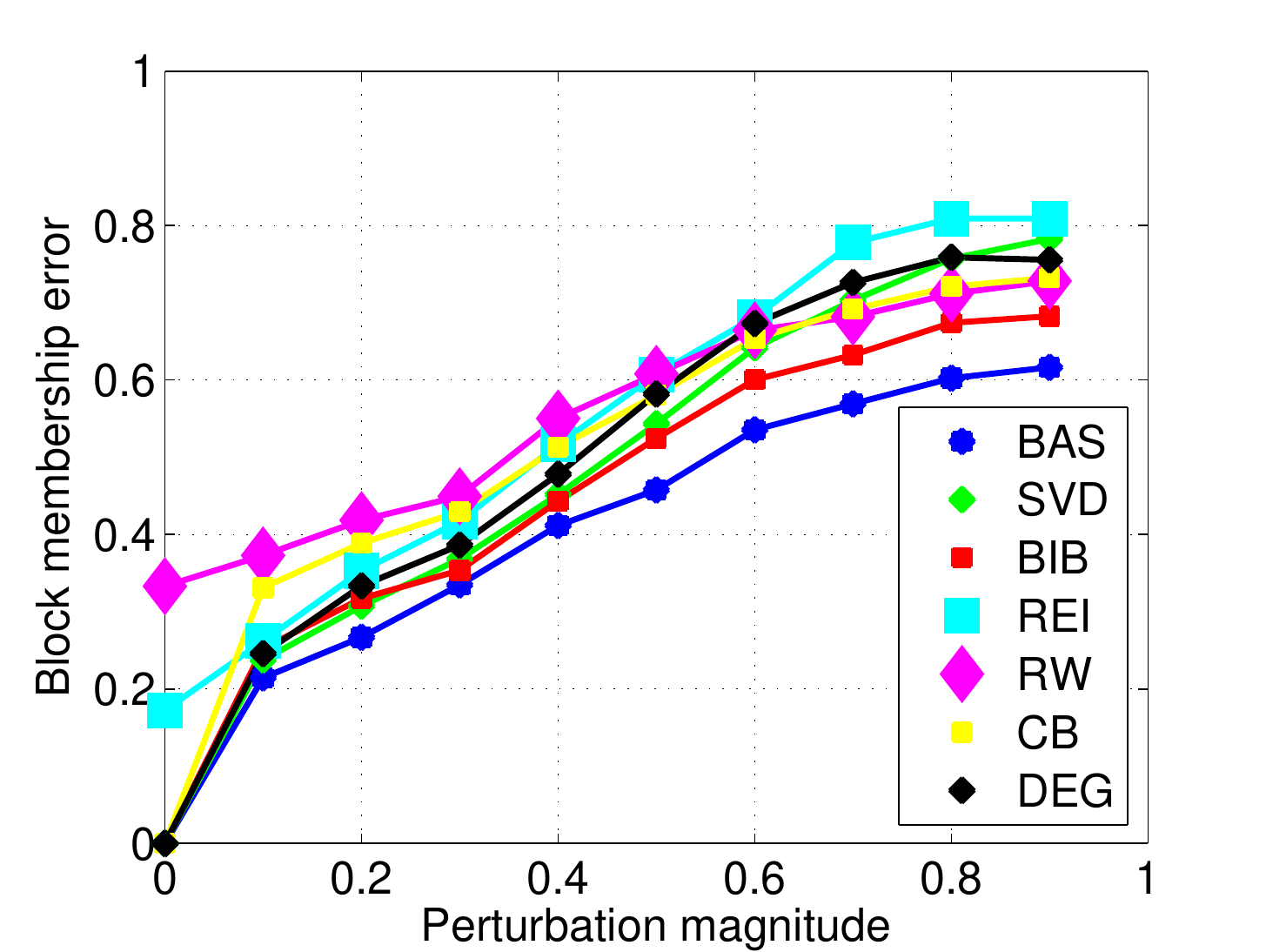}
\label{perturbedBlockAcycleNonUniform}}
\hfil
\subfloat[1-NMI]{\includegraphics[width=2.5in]{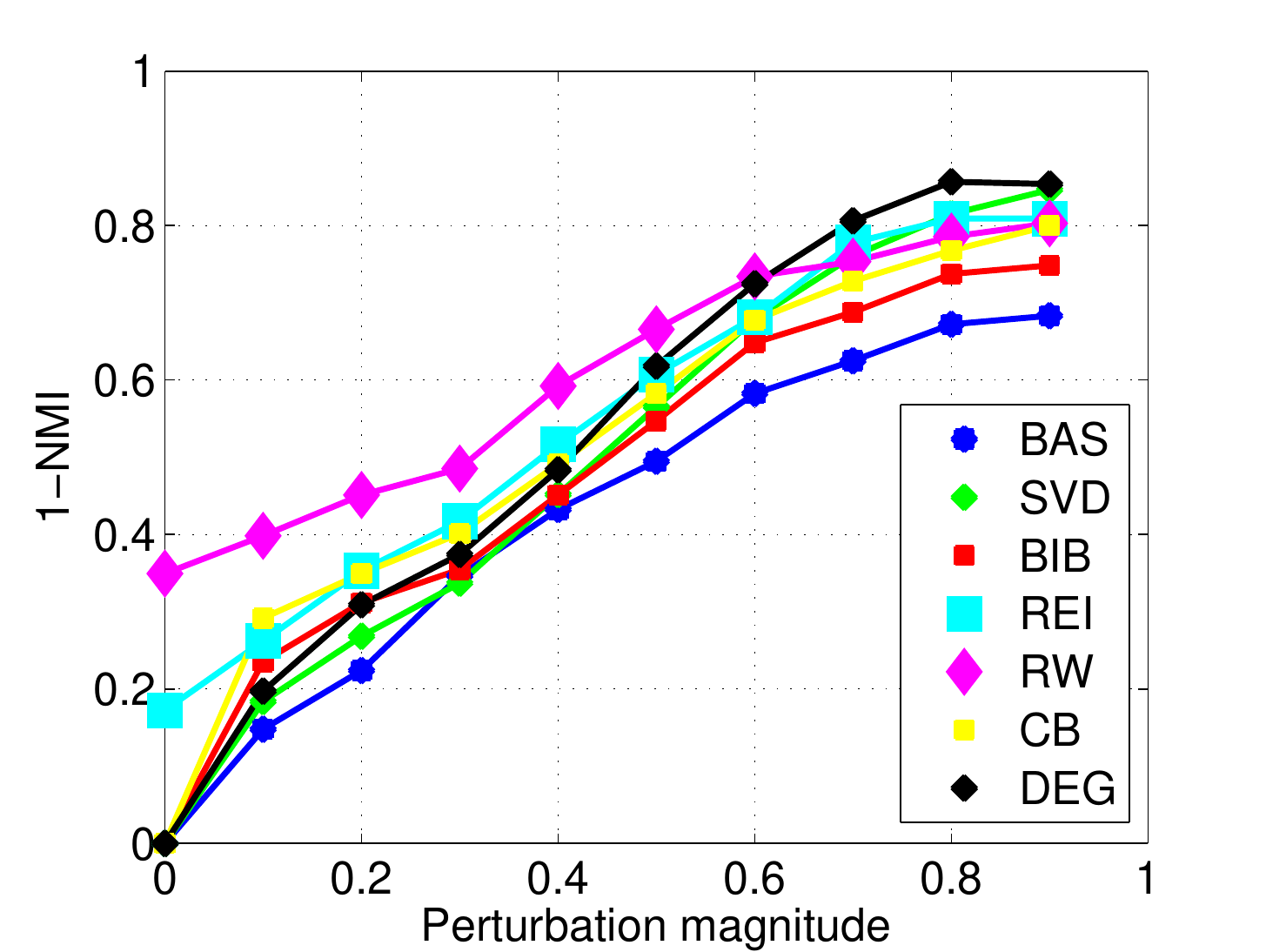}
\label{perturbedBlockAcycleUniform}}
\caption{Evolution of the block membership errors (left) and errors based on NMI (right) for our BAS algorithm and for SVD, BIB, RW, REI, CB and DEG clustering algorithms as a function of the magnitude of the perturbation (parameter $\epsilon$) on a block-acyclic graph of $n=1000$ nodes.}
\label{perturbedBlockAcycle}
\end{figure*}

\subsection{Experiment for BAS clustering algorithm}
Using the same framework as for block-cycles, in the first experiment, we generate an unperturbed block-acyclic graph $G$ with block membership function $\tau$ such that $[G,\tau]\sim SBM(k,\rho,P)$ with $k=8$
\begin{equation}
\begin{array}{rcl}
\rho &=& [0.18,0.2,0.05,0.2,0.14,0.04,0.07,0.13]\\
P_{st} &=& \left\lbrace\begin{array}{ll}
0.5 & \text{ if }s<t\\
0 & \text{ otherwise }
\end{array}\right.
\end{array}.
\end{equation} 
We generate a perturbing graph $\tilde{G}$ with block membership function $\tilde{\tau}$ such that $[\tilde{G},\tilde{\tau}]\sim SBM(k,\rho,\tilde{P})$ with parameter $\tilde{P}$ of the form $\tilde{P}=\epsilon Q$ where $\epsilon\in [0,1]$ and $Q\in [0,1]^{k\times k}$ with random entries in $[0,1]$. Figure \ref{perturbedBlockAcycle} displays the result for $n=1000$ nodes (target dimension $d$ for SVD clustering algorithm is chosen as the rank of matrix $P$ which is $k-1$). We observe that BAS clustering algorithm achieves the smallest block membership error and NMI-based error (or close to the smallest for perturbation of magnitude $0.3$).

For our second experiment, we generate two graphs following the same stochastic blockmodel described above $SBM(k,\rho,P)$, each containing $500$ nodes. Then we combine them using equation \ref{modeladv} with parameter $\alpha=0.1$. Table \ref{tableExperiment2} shows the result obtained (parameter $d$ of SVD clustering algorithm is again set to the rank of $P$, $k-1$ in this case). Our method recovers the block membership of nodes with a low error while other methods all produce block membership errors above $40\%$ as shown in table \ref{tableExperiment2}.
\begin{table}[!h]
\begin{center}
\begin{tabular}{|| c |c|c||}
   \hline
   Method & Block Membership Error & 1-NMI\\
   \hline
	BCS & $\mathbf{0.05}$ & $\mathbf{0.11}$\\
	SVD & $0.50$ & $0.46$\\
	BIB & $0.52$ & $0.49$\\
	REI & $0.53$ & $0.48$\\
	RW & $0.40$ & $0.38$\\
	CB & $0.56$ & $0.58$\\
	DEG & $0.52$ & $0.46$\\   
   \hline
\end{tabular}
\end{center}
\caption{Experiment based on two block-acyclic graphs of $500$ nodes generated by stochastic blockmodel $SBM(P,\rho,k)$ and combined to form a block-acyclic graph of $1000$ nodes based on equation \ref{modeladv} with $\alpha=0.1$.}
\label{tableExperiment2}
\end{table}

\section{Application to real-world networks}\label{ApplicationSection}
Pure block-cycles are rarely encountered in real-world networks. But block-acyclic networks (or graphs that are close to being block-acyclic) are common. In this section, we apply BAS clustering algorithm to two real-world networks: a trophic network and a network of Autonomous Systems.

Once the block partition $\tau: V\rightarrow \{1,...,k\}$ of nodes in graph $G=(V,E,W)$ is obtained, it is of interest to find the \textit{ranking} of blocks which can be interpreted as a topological order of blocks. We define a graph $G^B=(V^B,E^B)$ where $V^B=\{1,...,k\}$ and
\begin{equation}
(m,n)\in E^B \Leftrightarrow \underset{\substack{u:\tau(u)=m\\v:\tau(v)=n}}{\sum}W_{uv}-W_{vu}>0.
\end{equation}
If graph $G$ is indeed close to being block-acyclic, then graph $G^B$ is acyclic. The ranking of blocks is then obtained by computing a topological order of nodes in $G^B$ \cite{Cormen1990} which labels nodes in $G^B$ from $1$ to $k$ such that 
\begin{equation}
(i,j)\in E^B \Rightarrow i<j.
\end{equation}
Hence we relabel each block with its rank as given by the topological order in $G^B$. This ranking of blocks can be further used to improve the quality of the block partitioning returned by BAS clustering algorithm through a simple postprocessing step. Indeed, in real-world graphs, BAS clustering algorithm might confuse consecutive blocks in the hierarchy (for instance assigning a node to block $b-1$ or $b+1$ instead of its true block $b$) as the presence of perturbation causes the separation between corresponding clusters in the embedding space $\mathbb{R}^k$ to become fuzzy. Hence, we define a quality criterion $C_A$ measuring how close to block-acyclic the graph is, based on node partitioning $\tau$
\begin{equation}C_A=\frac{\vert \{(u,v)\in E:\tau(u)<\tau(v)\}\vert}{\vert E\vert}.
\end{equation}
The post-processing step considers each node $i$ and checks whether $C_A$ is increased by assigning $i$ to block $\tau(i)+1$ or $\tau(i)-1$ and changes the block membership of $i$ if this is the case. This process is repeated once for each node in a random order. This post-processing step causes a negligible increase in the time of computation of BAS clustering algorithm. But empirical analysis show that it improves the quality of the result for instance in the case of the trophic network presented below.

\begin{table*}[!h]
\begin{center}
\resizebox{\textwidth}{!}{
\begin{tabular}{||l|r||l|r||l|r||l|r||l|r||}
   \hline
   \multicolumn{2}{|c|}{\textbf{Block 1}} &\multicolumn{2}{|c|}{\textbf{Block 2}} &\multicolumn{2}{|c|}{\textbf{Block 3}} &\multicolumn{2}{|c|}{\textbf{Block 4}} & \multicolumn{2}{|c|}{\textbf{Block 5}}\\
   \hline
   Agents & Levels & Agents & Levels & Agents & Levels & Agents & Levels & Agents & Levels\\
   \hline
\textbf{Average}	&	\textbf{1.65}	&	\textbf{Average}	&	\textbf{3.13}	&	\textbf{Average}	&	\textbf{4.03} & \textbf{Average}	&	\textbf{4.69}	&	\textbf{Average}	&	\textbf{5.17}	\\
Input	&	0	&	Roots	&	1	&	Coral	&	3.43 & Rays	&	4.89	&	Sharks	&	5.12	\\
2um Spherical Phytoplankt	&	1	&	Water Cilitaes	&	2.9	&	Other Cnidaridae	&	4.26	& Bonefish	&	4.81	&	Tarpon	&	5.23	\\
Synedococcus	&	1	&	Acartia Tonsa	&	2.91	&	Echinoderma	&	3.7	&Lizardfish	&	4.98	&	Grouper	&	5.06	\\
Oscillatoria	&	1	&	Oithona nana	&	2.91	&	Lobster	&	4.58&Catfish	&	4.93	&	Mackerel	&	5.2	\\
Small Diatoms (<20um)	&	1	&	Paracalanus	&	2.91	&	Predatory Crabs	&	4.49	&Eels	&	4.93	&	Barracuda	&	5.21	\\
Big Diatoms (>20um)	&	1	&	Other Copepoda	&	3.36	&	Callinectus sapidus	&	4.49&Brotalus	&	4.85	&	Loon	&	5.21	\\
Dinoflagellates	&	1	&	Meroplankton	&	3.63	&	Stone Crab	&	4.31&Needlefish	&	4.72	&	Greeb	&	5.05	\\
Other Phytoplankton	&	1	&	Other Zooplankton	&	2.91	&	Sardines	&	4.18&Snook	&	4.64	&	Pelican	&	5.2	\\
Benthic Phytoplankton	&	1	&	Sponges	&	3	&	Anchovy	&	4.25 & Jacks	&	4.71	&	Comorant	&	5.18	\\
Thalassia	&	1	&	Bivalves	&	3	&	Bay Anchovy	&	3.95&Pompano	&	4.83	&	Big Herons and Egrets	&	5.14	\\
Thalassia	&	1	&	Bivalves	&	3	&	Bay Anchovy	&	3.95&Other Snapper	&	4.69	&	Predatory Ducks	&	5.14	\\
Halodule	&	1	&	Detritivorous Gastropods	&	4.13	&	Toadfish	&	4.85&Gray Snapper	&	4.78	&	Raptors	&	5.66	\\
Syringodium	&	1	&	Epiphytic Gastropods	&	2	&	Halfbeaks	&	3.26&Grunt	&	4.36	&	Crocodiles	&	5.39	\\
Drift Algae	&	1	&	Predatory Gastropods	&	5.12	&	Other Killifish	&	3.7&Scianids	&	4.63	&	Dolphin	&	5.27	\\
Epiphytes	&	1	&	Detritivorous Polychaetes	&	3.88	&	Goldspotted killifish	&	4.3&Spotted Seatrout	&	4.69	&	Water POC	&	5.04	\\
Free Bacteria	&	2.92	&	Predatory Polychaetes	&	4.44	&	Rainwater killifish	&	3.82&Red Drum	&	4.77	&	Benthic POC	&	4.55	\\
Water Flagellates	&	3.19	&	Suspension Feeding Polych	&	3.32	&	Silverside	&	3.95&Spadefish	&	4.58	&	Output	&	5.22	\\
Benthic Flagellates	&	3.77	&	Macrobenthos	&	4.13	&	Other Horsefish	&	3.97&Parrotfish	&	3.86	&	Respiration	&	5.19	\\
Benthic Ciliates	&	4.1	&	Benthic Crustaceans	&	3.88	&	Gulf Pipefish	&	4.06&Flatfish	&	4.58	&		&		\\
Meiofauna	&	4.35	&	Detritivorous Amphipods	&	3.88	&	Dwarf Seahorse	&	3.97&Filefishes	&	4.78	&		&		\\
	&		&	Herbivorous Amphipods	&	2.59	&	Mojarra	&	4.32&Puffer	&	4.77	&		&		\\
	&		&	Isopods	&	2	&	Porgy	&	4.26&Other Pelagic Fishes	&	4.74	&		&		\\
	&		&	Herbivorous Shrimp	&	2	&	Pinfish	&	3.82&Small Herons and Egrets	&	4.85	&		&		\\
	&		&	Predatory Shrimp	&	3.9	&	Mullet	&	3.86&Ibis	&	4.8	&		&		\\
	&		&	Pink Shrimp	&	3.43	&	Blennies	&	4.12&Roseate Spoonbill	&	4.91	&		&		\\
	&		&	Thor Floridanus	&	2	&	Code Goby	&	4.41&Herbivorous Ducks	&	4.19	&		&		\\
	&		&	Detritivorous Crabs	&	3.72	&	Clown Goby	&	4.41&Omnivorous Ducks	&	4.28	&		&		\\
	&		&	Omnivorous Crabs	&	3.79	&	Other Demersal Fishes	&	4.21&Gruiformes	&	4.91	&		&		\\
	&		&	Sailfin Molly	&	2	&	DOC	&	1.92&Small Shorebirds	&	4.93	&		&		\\
	&		&	Green Turtle	&	2	&		&		&Gulls and Terns	&	4.91	&		&		\\
&&&&&&Kingfisher	&	4.86	&		&		\\
&&&&&&Loggerhead Turtle	&	4.79	&		&		\\
&&&&&&Hawksbill Turtle	&	4.71	&		&		\\
&&&&&&Manatee	&	3.9	&		&		\\

   \hline
\end{tabular}}
\end{center}
\caption{Blocks returned by BAS clustering algorithm along with the trophic levels of each agent. Block 1 is at the bottom of the hierarchy, 5 is at the top.}
\label{trophicTable}
\end{table*}

\begin{figure}[!t]
\centering
\includegraphics[width=2.5in]{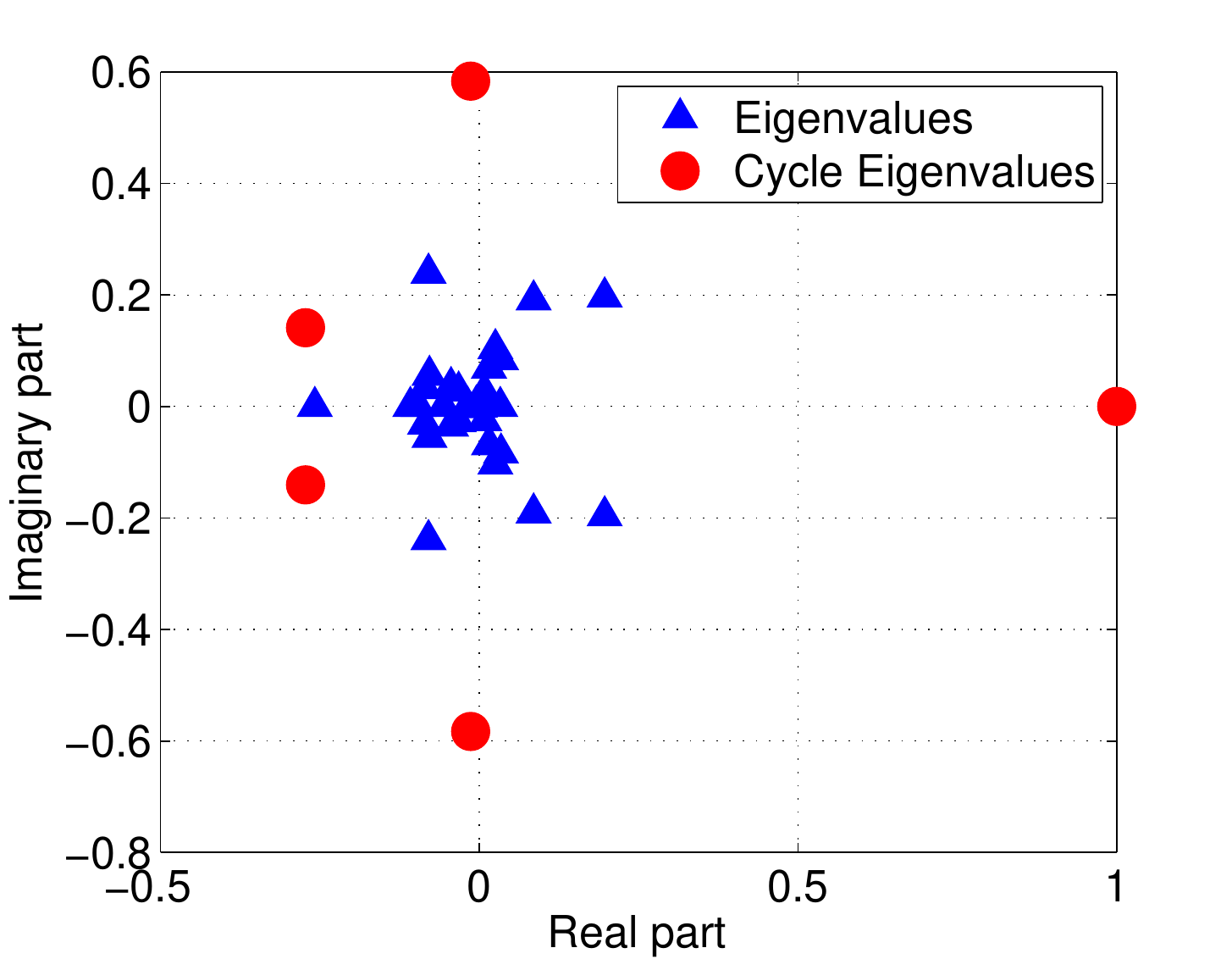}
\caption{Eigenvalues of the transition matrix of the trophic network. The five eigenvalues with largest modulus are the cycle eigenvalues.}
\label{eigenvalue_application2_trophic}
\end{figure}

\subsection{Finding clusters in trophic networks}
The basic idea of trophic networks is to represent all transfers of carbon in an ecosystem as a single directed graph \cite{Elton1927}, in which nodes are living beings or any other agent that stores carbon such as animals, carbon dioxide in the air, etc. A directed edge $(A,B)$ represents a steady transfer of carbon from $A$ to $B$, for instance a predator-prey relationship where $B$ is the predator and $A$ is the prey. Considering the trophic network of an isolated ecosystem, we should be able to partition nodes into groups such that nodes belonging to the same group have roughly the same groups of predators and groups of preys. Empirical observations lead to the following five categories \cite{Butz2004}: primary producers (plants and algae that produce organic material by photosynthesis), primary consumers (herbivores), secondary consumers (carnivores that hunt herbivores, such as badgers), tertiary consumers (carnivores hunting other carnivores, such as hawks) and top-level predators having no other predators (such as bears). Decomposers are sometimes associated to primary producers to form level 1 as neither of them hunts for its livelihood. Other types of "carbon holders" may be included in the network such as the atmosphere that receives carbon dioxide from all consumers and provides the producers with carbon. A popular way to estimate the role of a species in a food web is the computation of its trophic level. The trophic level can be viewed as the level at which a species can be found in a food chain. The trophic level $T_i$ of species $i$ is defined as $T_i=1+\sum_{j}T_jp_{ji}$
where $p_{ji}$ denotes the fraction of $j$ in the diet of $i$ \cite{Pauly2005}. Hence, if we define $P$ as the transition matrix of a food network with associated adjacency matrix $W$ then the vector $T$ of trophic levels is $T=(I-P^T)^+\mathbf{1}$ where $\mathbf{1}$ is the column vector of all ones and $(I-P^T)^+$ denotes the Moore-Penrose pseudoinverse of $(I-P^T)$ \cite{Penrose1955}.

\begin{table}[!h]
\begin{center}
\resizebox{\textwidth}{!}{
\begin{tabular}{|c|c|c|c|c|c|}
   \hline
   & \begin{tabular}{@{}c@{}}Primary\\ Producers\end{tabular}
   & \begin{tabular}{@{}c@{}}Primary\\ Consumers\end{tabular}
   & \begin{tabular}{@{}c@{}}Secondary\\ Consumers\end{tabular}
   & \begin{tabular}{@{}c@{}}Tertiary\\ Consumers\end{tabular}
   & \begin{tabular}{@{}c@{}}Top-Level\\ Predators\end{tabular}
   \\
   \hline
	\begin{tabular}{@{}c@{}}Primary\\ Producers\end{tabular}& $6$ & $\mathbf{25}$ & $\mathbf{16}$ & $\mathbf{9}$ & $\mathbf{9}$\\
	   \hline
	\begin{tabular}{@{}c@{}}Primary\\ Consumers\end{tabular}& $0$ & $7$ & $\mathbf{47}$ & $\mathbf{41}$ & $\mathbf{23}$\\
	   \hline
	\begin{tabular}{@{}c@{}}Secondary\\ Consumers\end{tabular}& $0.4$ & $2$ & $9$ & $\mathbf{29}$ & $\mathbf{43}$ \\   			   \hline
	\begin{tabular}{@{}c@{}}Tertiary\\ Consumers\end{tabular}& $0$ & $0.3$ & $6$ & $2$ & $\mathbf{26}$\\
       \hline
	\begin{tabular}{@{}c@{}}Top-Level\\ Predators\end{tabular}& $0.6$ & $2$ & $4$ & $0$ & $7$\\ 
       \hline
\end{tabular}}
\end{center}
\caption{Percentage of directed edges between blocks: entry $(i,j)$ equals $100 E(i,j)/(S(i)S(j))$ where $E(i,j)$ is the total number of directed edges from block $i$ to block $j$ and $S(i)$ and $S(j)$ are the number of nodes in blocks $i$ and $j$ respectively. Upper triangular part of the matrix appears in bold.}
\label{adjmatrix_trophic}
\end{table}

It is clear that the adjacency matrix of a food web has a block upper triangular shape and intuition suggests that the corresponding graph is block-acyclic. Our goal is to use BAS clustering algorithm to partition the agents of a food web into five groups corresponding to the five categories described previously. Then, we check if the partitioning of nodes that we obtain is consistent with the trophic levels.
We apply BAS clustering algorithm to a trophic network of Florida Bay \cite{Ulanowicz2005} which consists of $128$ species or other relevant agents in the trophic network. Having information about the feeding relationships in the food web, we build a directed graph $G=(V,E)$ where a directed unweighted edge connects node $A$ to node $B$ if $B$ feeds on $A$. Figure \ref{eigenvalue_application2_trophic} shows the eigenvalues of the transition matrix and the cycle eigenvalues used by BAS clustering algorithm. We observe that all but a few eigenvalues are close to zero. As we partition nodes into five blocks, we select the five eigenvalues with largest norm.

Table \ref{trophicTable} shows the five blocks extracted by BAS clustering algorithm. The table also includes the trophic levels of each agent and the average trophic level of each cluster. We observe that the block partitioning we obtain is consistent with the empirical classification into five categories described above:
\begin{itemize}
\item \textbf{block 1:} mostly autotrophs and bacteria,
\item \textbf{block 2:} herbivorous and small omnivorous (such as gastropods and shrimps),
\item \textbf{block 3:} larger omnivorous and carnivorous (Killifish, crabs, lobsters...),
\item \textbf{block 4:} carnivorous (kingfishers, catfish, snappers...),
\item \textbf{block 5:} top-level predators (sharks, cormorants...).
\end{itemize}

Table \ref{adjmatrix_trophic} shows the percentage of directed edges observed between blocks obtained by BAS clustering algorithm which exhibits the block upper triangular shape of the adjacency matrix of the trophic network. Finally, we check that our classification is globally consistent with trophic levels. Let us denote by $V$ the set of $n=128$ agents, $\tau:V\rightarrow\{1,...,5\}$ the function assigning each agent to a block and $l:V\rightarrow\mathbb{R}^+$ the trophic levels such that $l(x)$ is the trophic level of agent $x$. We compute the \textit{inversion error}
\begin{equation}
\frac{2}{n(n-1)}|\{(i,j)\in V\times V\text{ : }(\tau(i)-\tau(j))(l(i)-l(j))<0\}|
\end{equation}
which computes the proportion of pairs $i,j$ for which block memberships are inconsistent with trophic levels. We obtain an inversion error of $7\%$ which means that the partitioning is consistent with the ground truth trophic levels. In comparison, the partitioning obtained by applying SVD clustering algorithm produces an inversion error of $25\%$. This confirms that our algorithm is more efficient for detecting blocks in block-acyclic networks. As said before, the failure of SVD clustering algorithm is due to the inability of stochastic blockmodels to capture the structure of trophic networks.

\subsection{Network of Autonomous Systems}
As a second application, we apply BAS clustering algorithm to an internet-based network. Autonomous Systems are sets of computers sharing a common routing protocol which roughly corresponds to computers that get their internet connection from the same Internet Service Provider (ISP). Based on their importance, ISPs can be partitioned in tiers: ISPs in lower tiers sign business agreements with higher tiers' ISPs for data transit. Hence the network of money transfers between ISPs has a hierarchical structure\cite{Center2013}.

We consider a graph $G=(V,E,W)$ in which vertices are ISPs and an edge $(u,v)$ represents a money transfer from ISP $u$ to ISP $v$ during a certain time span. We use the dataset published by the \textit{Center for Applied Internet Data Analysis} \cite{Center2013} for November 11, 2013 which involves $45,427$ ISPs (nodes) and 230,194 connections (edges). Unfortunately, business agreements between ISPs are often kept secret, hence these relationships are inferred from Border Gateway Protocol data using a heuristic method from \cite{Luckie2013}. To eliminate bidirectional connections between ISPs belonging to the same tier, we keep the asymmetric part of $W$ only: $W\leftarrow (W-W^T)_+$. The resulting graph is expected to be block-acyclic \cite{Van2015} and the ranks of vertices in the block-acyclic structure should reflect the partitioning of ISPs into tiers. Previous research suggested to partition ISPs into $3$ tiers with tier $1$ containing the most important ISPs in terms of data transit \cite{Winther2006}. The spectrum of the transition matrix of the graph is shown in figure \ref{eigenvalue_application3}. We observe that there are indeed three eigenvalues with a significantly larger modulus. Hence we apply BAS clustering algorithm for the extraction of $k=3$ blocks. Figure \ref{application3_graph} shows the result of the partitioning and the number of connections between each block in the original network (of adjacency matrix $W$). This partitioning reveals the block-acyclic structure of the AS network. 

In order to assess the quality of the partitioning, we compare it to two available measures of the importance of Autonomous Systems. We first consider the transit degree of ISPs which measures the number of ISPs for which a given ISP provides a transit of data \cite{Center2013}. Considering the partitioning computed by BAS clustering algorithm, the average transit degree is $0.31$ in block $1$, $5.16$ in block $2$ and $16.49$ in block $3$. Moreover the inversion error between transit degrees and ranks computed by BAS clustering algorithm is $6\%$ which confirms the consistency of the partitioning with transit degrees. Secondly, we check the consistency of blocks with the grouping into tiers. There is no unique partitioning of ISPs in tiers and most of the ones that are available include ISPs in tier $1$ only. Hence we only check the coherence of our block partitioning with the most common version of tier $1$ which includes sixteen top ISPs in the world among which AT\&T (U.S.), Deutsche Telekom (Germany), etc. All of these tier $1$ ISPs are classified in block $3$ of the block-acyclic network. The blocks discovered by our BAS clustering algorithm are thus somewhat coherent with the traditional partitioning into tiers. However both classifications are not expected to be equivalent since BAS clustering algorithm is only based on the graph topology while tiers also take additional information into account such as the ownership of an international fiber optic network, etc. 

\begin{figure}[!t]
\centering
\includegraphics[width=2.5in]{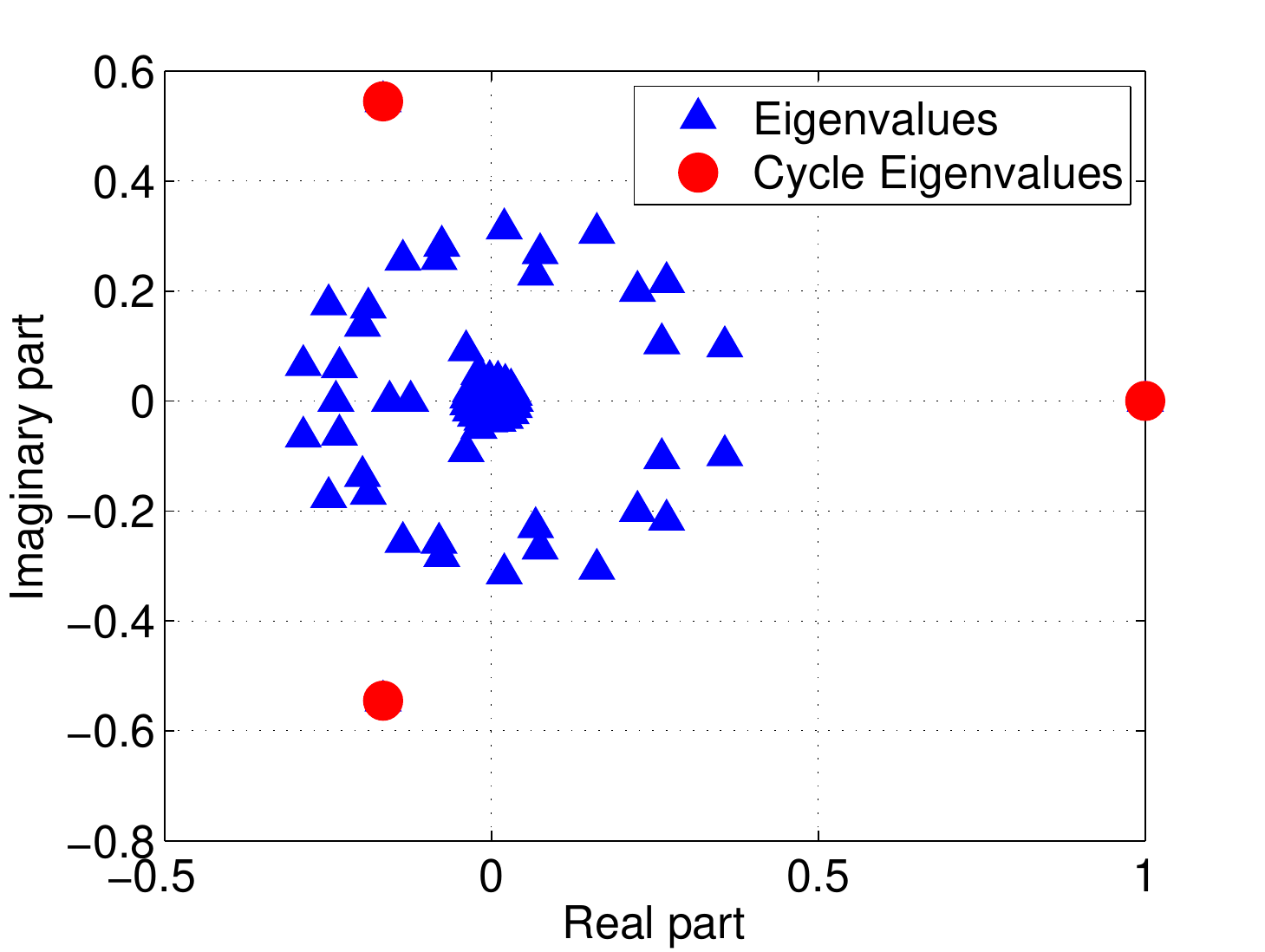}
\caption{$200$ top eigenvalues of the transition matrix of the AS network. The three cycle eigenvalues (with largest modulus) used by BAS clustering algorithm are circle-shaped. }
\label{eigenvalue_application3}
\end{figure} 

\begin{figure}[!t]
\centering
\includegraphics[width=2.5in]{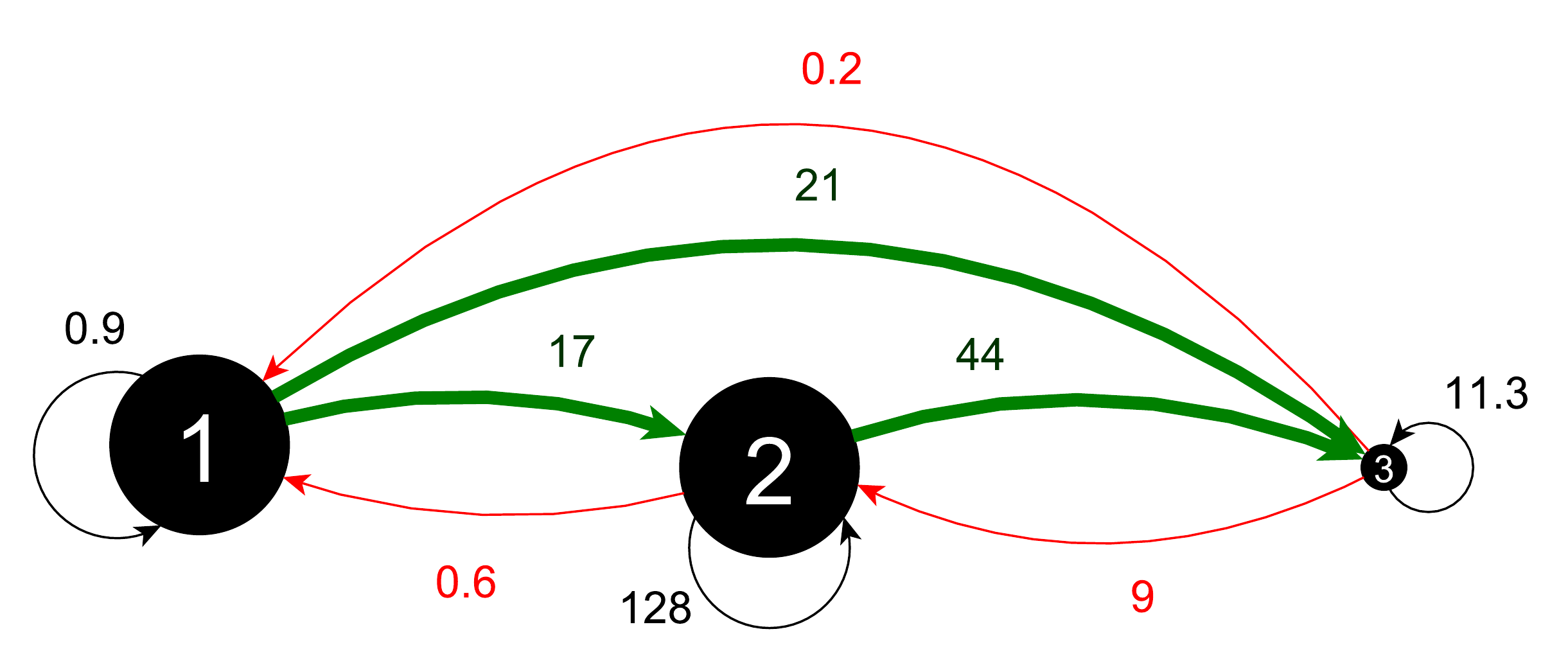}
\caption{Graph of the partitioning of nodes of the AS network into three blocks. Number of nodes in blocks 1, 2 and 3 are respectively 18,171, 20,896 and 6,360. The label of each edge represents the number of connections from one block to another block (in thousands).}
\label{application3_graph}
\end{figure}

\section{Conclusion}
In this paper we proposed two new algorithms to detect blocks of nodes in graphs with a block-cyclic or a block-acyclic structure. These algorithms are based on the computation of complex eigenvectors and eigenvalues of the transition matrix associated to the graph. We show that the algorithms succeed in detecting such blocks and they outperform other methods that are theoretically able to extract clusters from block-cyclic or block-acyclic graphs. As we mentioned, we seek specific structural patterns (cyclic or acyclic) in graphs but we make no assumption on the degrees of nodes or the number of connections which makes our tailored approach more efficient and cheaper than methods that seek more general patterns such as algorithms based stochastic blockmodels. Moreover, the time complexity of BCS and BAS clustering algorithms is linear in the number of edges in the graph which is the same as for other spectral methods such as traditional spectral clustering. We also show that BAS clustering algorithm provides a better understanding of a real-world database: it extracts trophic levels from a food web and it exhibits the hierarchical structure of a worldwide network of Autonomous Systems. As we mentioned, the lower quality of the partitioning provided by algorithms such as the ones based on stochastic blockmodels in block-cyclic and block-acyclic graphs is partly due to some assumptions of regularity within blocks which are not necessarily verified in some real-world networks such as food webs. In contrast, our algorithms detect block-cyclic and block-acyclic structures regardless of other regularity properties.

Future work may include the development of an automatic method for choosing a suitable number $k$ of blocks which could be based on an inspection either of the spectrum of the transition matrix or of the graph itself. Another research direction would be to generalize this framework for the detection of clusters with more complex patterns of connection. For instance, we may generalize our algorithm to graphs in which a subset of nodes forms a block-cyclic structure (or a block-acyclic structure). If this block-cyclic component is weakly connected to the rest of the graph, then a straightforward extension of BCS clustering algorithm would detect it successfully. Finally BCS and BAS clustering algorithms could be applied to other real-world databases. In his book "The Human Group" \cite{Homans1950} published in 1950, the American sociologist George Homans states that small social groups tend to form block-symmetric-acyclic networks, namely block-acyclic graphs with additional bidirectional links within blocks. It would be interesting to use BAS clustering algorithm to check if this assumption is verified for large web-based social networks involving directed connections such as Twitter for instance.

\section*{Acknowledgement}
This work presents research results of the Belgian Network DYSCO (Dynamical Systems, Control, and Optimisation), funded by the Interuniversity Attraction Poles Programme, initiated by the Belgian State, Science Policy Office, and the ARC (Action de Recherche Concerte) on Mining and Optimization of Big Data Models funded by the Wallonia-Brussels Federation.

\section*{Appendix}
\appendix
The proofs of theorems \ref{theoremCycle} and \ref{theoremPerturbation} are provided below.

\begin{theorem}[Existence of cycle eigenvalues]~\\
Let $G=(V,E,W)$ be a block-cycle with $k$ blocks $V_1,...,V_k$ such that $d_i^{out}>0$ for all $i\in V$. Then $\lambda_l=e^{-2\pi i\frac{l}{k}}\in spec(P)$ for all $0\leq l\leq k-1$, namely there are $k$ eigenvalues located on a circle centered at the origin and with radius $1$ in the complex plane. The eigenvector associated to the eigenvalue $e^{-2\pi i\frac{l}{k}}$ is
$$
u^l_j=\left\{
\begin{array}{ll}
e^{2\pi i\frac{lk}{k}} & j\in V_1\\
e^{2\pi i\frac{l(k-1)}{k}} & j\in V_2\\
\vdots & \\
e^{2\pi i\frac{l}{k}} & j\in V_k
\end{array}.\right.
$$ 
Moreover, if $G$ is strongly connected, then the eigenvalues $\lambda_0,...,\lambda_{k-1}$ have multiplicity $1$ and all other eigenvalues of $P$ have a modulus strictly lower than $1$.
\end{theorem}
\begin{proof}
For the $l$-th eigenvalue $\lambda_l=e^{-2\pi i\frac{l}{k}}$, we consider the following eigenvector:
$$
u^l_j=\left\{
\begin{array}{ll}
e^{2\pi i\frac{lk}{k}} & j\in V_1\\
e^{2\pi i\frac{l(k-1)}{k}} & j\in V_2\\
\vdots & \\
e^{2\pi i\frac{l}{k}} & j\in V_k
\end{array}.\right.
$$
Then, $\forall j\in V_s$ with $1\leq s\leq k-1$:
$$
[Pu^l]_j=\frac{1}{d^{out}_j}\sum_{r\in V_{s+1}}w_{jr}u^l_r.
$$
If $j\in V_s$ and $r\in V_{s+1;}$, clearly, $u^l_r=e^{-2\pi i\frac{l}{k}}u^l_j$, hence
$$
\begin{array}{rcl}
[Pu^l]_j &=& \frac{1}{d^{out}_j}\sum_{r\in V_{s+1}}w_{jr}u^l_r\\
&=& \frac{1}{d^{out}_j}\sum_{r\in V_{s+1}}w_{jr} e^{-2\pi i\frac{l}{k}}u^l_j\\
&=& e^{-2\pi i\frac{l}{k}}u^l_j\frac{1}{d^{out}_j}\sum_{r\in V_{s+1}}w_{jr}\\
&=& e^{-2\pi i\frac{l}{k}}u^l_j
\end{array}
$$
Let us assume now that $G$ is strongly connected. To show that eigenvalues $\lambda_0,...,\lambda_{k-1}$ have multiplicity $1$, we refer to Perron Frobenius theorem \cite{Seneta2006}. As $G$ is strongly connected, $P$ is irreducible and the definition of block-cycle implies that the period of $G$ is $k$. Hence, $P$ has exactly $k$ eigenvalues, each of multiplicity $1$ and having modulus $1$. These eigenvalues are necessarily the $k$ cycle eigenvalues $\lambda_0,...,\lambda_{k-1}$.
\end{proof}

\begin{theorem}[Perturbation of Cycle Eigenvalues]
~\\
Let $G=(V,E,W)$ be a strongly connected block-cycle with $k$ blocks $V_1,...,V_k$ such that $d_i^{out}>0$ for all $i\in V$, let $\lambda_0,...,\lambda_{k-1}$ be the $k$ cycle eigenvalue and $u^l$, $y^l$ be the corresponding right and left eigenvectors. Let the $\hat{G}=(V,\hat{E},\hat{W})$ be a perturbed version of $G$ formed by appending positively weighted edges to $G$ except self-loops. Let $P$ and $\hat{P}$ denote the transition matrices of $G$ and $\hat{G}$ respectively. We define the quantities
\begin{equation}\label{deltaKetc2}
\begin{array}{rcl}
\sigma &=&\underset{(i,j)\in\hat{E}}{\max}\text{ }\frac{\hat{d}_j^{in}}{\hat{d}_i^{out}}\\
\rho &=&\underset{i}{\max}\text{ }\frac{\hat{d}_i^{out}-d_i^{out}}{d_i^{out}}
\end{array}
\end{equation}

where $d^{in}_i$, $d^{out}_i$, $\hat{d}^{in}_i$ and $\hat{d}^{out}_i$ represent the in-degree and out-degree of $i$-th node in $G$ and $\hat{G}$ respectively.

Then,
\begin{enumerate}
\item for any cycle eigenvalue $\lambda_l\in spec(P)$, there exists an eigenvalue $\hat{\lambda}_l$ so that

$$
\left\vert\hat{\lambda}_l-\lambda_l\right\vert \leq \sqrt{2n}\Vert f\Vert_2\sigma^{\frac{1}{2}}\rho^{\frac{1}{2}}+\mathcal{O}\left(\sigma\rho\right)
$$
where $f$ is the Perron eigenvector of $G$, namely the left eigenvector of the transition matrix of $G$ associated to eigenvalue $1$ with $\Vert f\Vert_1=1$,

\item there exists an eigenvector $\hat{u}^l$ of $\hat{P}$ associated to eigenvalue $\hat{\lambda}^l$ verifying
$$ 
\Vert \hat{u}^l-u^l\Vert_2\leq \sqrt{2}\Vert (\lambda^lI-P)^{\#}\Vert_2 \sigma^{\frac{1}{2}}\rho^{\frac{1}{2}}+\mathcal{O}\left(\sigma\rho\right)
$$
where $u^l$ is the eigenvector of $P$ associated to eigenvalue $\lambda^l$ and $(\lambda^lI-P)^{\#}$ denotes the Drazin generalized inverse of $(\lambda^lI-P)$.
\end{enumerate}

\end{theorem}

\begin{proof}
We restrict ourselves to a perturbation that only consists of edges not already in $E$ as appending an edge that is already in $E$ amounts to increase the weight of this edge in $G$ which does not affect the block-cyclic structure of $G$. The transition matrix $\hat{P}$ can be expressed in the form $\hat{P}=P+R$ with $R$ given by
$$
R_{ij}=\left\lbrace
\begin{array}{ll}
\frac{\hat{w}_{ij}}{\hat{d}_i^{out}} & \text{ if }(i,j)\in \hat{E}\setminus E\\
\frac{w_{ij}}{\hat{d_i}^{out}}-\frac{w_{ij}}{d_i^{out}} & \text{ if }(i,j)\in E\cap \hat{E}\\
0 & \text{ otherwise.}
\end{array}\right.
$$
To prove claim $1$, we recall that cycle eigenvalues $\{\lambda_l,0\leq l\leq k-1\}$ have multiplicity $1$ from theorem \ref{theoremCycle}. Hence, from theorem 2.3 page 183 of \cite{Stewart1990}, for each cycle eigenvalue $\lambda_l$, there exists an eigenvalue $\hat{\lambda}_l\in spec(\hat{P})$ such that
$$
\hat{\lambda}_l=\lambda_l +\frac{(y_l)^HRu_l}{(y_l)^Hu_l}+\mathcal{O}(\Vert R\Vert_2^2)
$$
where $u_l$ and $y_l$ are the right and left eigenvectors of $P$ associated to cycle eigenvalue $\lambda_l$ and such that $\Vert u_l\Vert_2=\Vert y_l\Vert_2=1$.
\begin{equation}\label{modeleqn}
\left\vert\hat{\lambda}_l-\lambda_l\right\vert \leq \frac{\left\vert(y_l)^HRu_l\right\vert}{\left\vert(y_l)^Hu_l\right\vert}+\mathcal{O}(\Vert R\Vert_2^2)
\end{equation}
The denominator of the first order term can be expressed in the following way. One can show that the normalized right and left eigenvectors associated to cycle eigenvalues are
$$
\begin{array}{rclcl}
u_j^l&=&\frac{1}{\sqrt{n}} e^{il(k-r+1)\frac{2\pi}{k}}&\text{ for }&j\in V_r\\
y_j^l&=&y_j^0e^{-il(r-1)\frac{2\pi}{k}}&\text{ for }&j\in V_r
\end{array}
$$
where $y^0$ is the left eigenvector of the transition matrix of $G$ associated to eigenvalue $1$ (Perron vector) with $\Vert y^0\Vert_2=1$. As $y^0$ is real and non-negative by Perron Frobenius theorem,
$$
\begin{array}{rcl}
(y^l)^Hu^l &=& \sum_j (y_j^l)^*x_j^l\\
&=& \sum_j \frac{y_j^0}{\sqrt{n}}e^{il(r-1)\frac{2\pi}{k}}e^{il(k-r+1)\frac{2\pi}{k}}\\
&=& \sum_j e^{2\pi li}\frac{y_j^0}{\sqrt{n}}\\
&=& \sum_j \frac{y_j^0}{\sqrt{n}}.
\end{array}
$$
Hence
$$
|(y^l)^Hu^l|=\frac{\Vert y^0\Vert_1}{\sqrt{n}}.
$$
or if we consider the vector of stationary distribution $f$, we have $y^0=\frac{f}{\Vert f\Vert_2}$ and
$$
|(y^l)^Hu^l|=\frac{1}{\Vert f\Vert_2\sqrt{n}}.
$$
Regarding the numerator, we have
$$
\left\vert(y_l)^HRu_l\right\vert\leq \Vert y_l\Vert_2 \Vert R\Vert_2 \Vert u_l\Vert_2=\Vert R\Vert_2 
$$
where $\Vert R\Vert_2=\sqrt{\lambda_{\max}(R^TR)}$. By Gershgorin circle theorem, we have
\begin{equation}\label{lxkwjclkx}
\lambda_{\max}(R^TR) \leq \underset{t}{\max}\sum_j\sum_s \left\vert\frac{\hat{w}_{st}}{\hat{d}^{out}_s}-\frac{w_{st}}{d^{out}_s}\right\vert \left\vert\frac{\hat{w}_{sj}}{\hat{d}^{out}_s}-\frac{w_{sj}}{d^{out}_s}\right\vert
\end{equation}
We simplify the right member of the equation:
\begin{equation}
\sum_j\sum_s \left\vert\frac{\hat{w}_{st}}{\hat{d}^{out}_s}-\frac{w_{st}}{d^{out}_s}\right\vert \left\vert\frac{\hat{w}_{sj}}{\hat{d}^{out}_s}-\frac{w_{sj}}{d^{out}_s}\right\vert
=\sum_s \left\vert\frac{\hat{w}_{st}}{\hat{d}^{out}_s}-\frac{w_{st}}{d^{out}_s}\right\vert \sum_j \left\vert\frac{\hat{w}_{sj}}{\hat{d}^{out}_s}-\frac{w_{sj}}{d^{out}_s}\right\vert
\label{lskljf}
\end{equation}
in which
\begin{equation}
\begin{array}{rcl}
\sum_j \left\vert\frac{\hat{w}_{sj}}{\hat{d}^{out}_s}-\frac{w_{sj}}{d^{out}_s}\right\vert &=& \underset{j:w_{sj}=0}{\sum}\frac{\hat{w}_{sj}}{\hat{d}^{out}_s}+\underset{j:w_{sj}\neq 0}{\sum}w_{sj}(\frac{1}{d^{out}_s}-\frac{1}{\hat{d}^{out}_s})\\
&=& \frac{\hat{d}^{out}_s-d^{out}_s}{\hat{d}^{out}_s}+1-\frac{d_s^{out}}{\hat{d}^{out}_s}\\
&=& 2(1-\frac{d_s^{out}}{\hat{d}_s^{out}})\\
&=& 2\frac{1}{1+\frac{d_s^{out}}{\hat{d}_s^{out}-d_s^{out}}}\\
&\leq& 2\frac{\rho}{\rho +1}
\end{array}
\label{lskfjsl}
\end{equation}
Moreover, we have
\begin{equation}
\begin{array}{rcl}
\sum_s \left\vert\frac{\hat{w}_{st}}{\hat{d}^{out}_s}-\frac{w_{st}}{d^{out}_s}\right\vert
& = & \sum_{s:w_{st}=0}\frac{\hat{w}_{st}}{\hat{d}_s^{out}}+\sum_{s:w_{st}\neq 0}w_{st}(\frac{1}{d_s^{out}}-\frac{1}{\hat{d}_s^{out}})\\
&\leq& \sum_{s:w_{st}=0}\frac{\hat{w}_{st}}{\hat{d}_s^{out}}+\sum_{s:w_{st}\neq 0}\frac{w_{st}}{d_s^{out}}\\
&\leq& \sum_s \frac{\hat{w}_{st}}{d_s^{out}}\\
&=&  \sum_s \frac{\hat{w}_{st}}{\hat{d}_s^{out}}\frac{\hat{d}_s^{out}}{d_s^{out}}\\
&\leq & \sigma(\rho+1)
\end{array}
\end{equation}
Putting equation \ref{lskfjsl} back into \ref{lskljf}, we have
\begin{equation}
\begin{array}{rcl}
\sum_s \left\vert\frac{\hat{w}_{st}}{\hat{d}^{out}_s}-\frac{w_{st}}{d^{out}_s}\right\vert \sum_j \left\vert\frac{\hat{w}_{sj}}{\hat{d}^{out}_s}-\frac{w_{sj}}{d^{out}_s}\right\vert &\leq& 2\sigma(\rho+1)\frac{\rho}{\rho +1}\\
&=& 2\sigma\rho
\end{array}
\label{jlfklsk}
\end{equation}
Hence, the first order term of equation \ref{modeleqn} can be written as
$$
\frac{\left\vert(y_l)^HRu_l\right\vert}{\left\vert(y_l)^Hu_l\right\vert}=\sqrt{2n}\Vert f\Vert_2\sigma^{\frac{1}{2}}\rho^{\frac{1}{2}}
$$

Finally, the upper bound on the perturbation of the cycle eigenvalues
$$
\left\vert\hat{\lambda}_l-\lambda_l\right\vert \leq \sqrt{2n}\Vert f\Vert_2\sigma^{\frac{1}{2}}\rho^{\frac{1}{2}}+\mathcal{O}\left(\sigma\rho\right)
$$

To prove the second claim, as cycle eigenvalue $\lambda_l$ has multiplicity $1$, from theorem 2.8 p. 238 of \cite{Stewart1990}, there exists an eigenvector $\hat{u}^l$ of $\hat{P}$ associated to eigenvalue $\hat{\lambda}^l$ verifying
$$ 
\Vert \hat{u}^l-u^l\Vert_2\leq \Vert (\lambda^lI-P)^{\#}\Vert_2 \Vert R\Vert_2+\mathcal{O}\left(\Vert R\Vert_2^2\right)
$$
where $u^l$ is the eigenvector of $P$ associated to eigenvalue $\lambda^l$ and $(\lambda^lI-P)^{\#}$ denotes the Drazin generalized inverse of $(\lambda^lI-P)$. From equation \ref{jlfklsk}, this expression becomes
$$ 
\Vert \hat{u}^l-u^l\Vert_2\leq \sqrt{2}\Vert (\lambda^lI-P)^{\#}\Vert_2 \sigma^{\frac{1}{2}}\rho^{\frac{1}{2}}+\mathcal{O}\left(\sigma \rho\right)
$$
\end{proof}

\bibliography{mybibliography}{}
\bibliographystyle{unsrt}

\end{document}